\documentclass [journal,onecolumn,12pt]{IEEEtran}
\usepackage{amsfonts,amsmath,amssymb}
\usepackage{indentfirst, setspace}
\usepackage{url,float}
\usepackage{color}
\usepackage[a4paper,ignoreall]{geometry}
\usepackage{longtable}
\usepackage{graphicx}
\usepackage[a4paper,ignoreall]{geometry}
\usepackage{multicol}
\usepackage{stfloats}
\usepackage{enumerate}
\usepackage{amsthm}
\usepackage{booktabs}
\usepackage{threeparttable}
\usepackage{mathrsfs}
\usepackage{multirow}
\usepackage{amssymb}
\usepackage{verbatim}
\usepackage{cases}
\usepackage{makecell}
\usepackage[square, comma, sort&compress, numbers]{natbib}
\usepackage[overload]{empheq}

\def\qu#1 {\fbox {\footnote {\ }}\ \footnotetext { From Qu: {\color{red}#1}}}
\def\hqu#1 {}
\def\kq#1 {\fbox {\footnote {\ }}\ \footnotetext { From KangQuan: {\color{blue}#1}}}
\def\hkq#1 {}

\usepackage{longtable}

\newtheorem{Th}{Theorem}[section]

\newtheorem{Lemma}[Th]{Lemma}
\newtheorem{Def}[Th]{Definition}

\newtheorem{Rem}[Th]{Remark}

\newcommand{\tr}{{\rm Tr}}

\newcommand{\rank}{{\rm rank}}

\newcommand{\gf}{{\mathbb F}}

\newcommand{\supp}{{\rm Supp}}

\newcommand{\bu}{{\bf u}}
\newcommand{\bv}{{\bf v}}
\newcommand{\wt}{{\rm wt}}

\makeatletter
\newcommand{\figcaption}{\def\@captype{figure}\caption}
\newcommand{\tabcaption}{\def\@captype{table}\caption}
\makeatother

\begin{document}
	\title{ New constructions of $2$-to-$1$ mappings over $\gf_{2^n}$ and their applications to binary linear codes}
	\author{{Yaqin Li, Kangquan Li, Qiancheng Zhang}\\
	\thanks{  Yaqin Li, Kangquan Li and Qiancheng Zhang are with the College of Science,
		National University of Defense Technology, Changsha, 410073, China.
		E-mail:   liyaqin20@nudt.edu.cn, likangquan11@nudt.edu.cn, zhangqiancheng20@nudt.edu.cn.
		This work is supported by the National Key R\&D Program of China (No. 2024YFA1013000) and the National Natural Science Foundation of China  (No. 62202476 and 62172427). 
		{\emph{(Corresponding author: Kangquan Li)}}}
	}
	\maketitle{}

\begin{abstract}
	The $2$-to-$1$ mapping over finite fields has a wide range of applications, including combinatorial mathematics and coding theory. Thus, constructions of $2$-to-$1$ mappings have attracted considerable attention recently. Based on summarizing the existing construction results of all $2$-to-$1$ mappings over finite fields with even characteristic, this article first applies the generalized switching method to the study of $2$-to-$1$ mappings, that is, to construct $2$-to-$1$ mappings over the finite field $\mathbb{F}_{q^l}$ with $F(x)=G(x)+{\rm Tr}_{q^l/q}(R(x))$, where $G$ is a monomial and $R$ is a monomial or binomial. Using the properties of Dickson polynomial theory and the complete characterization of low-degree equations, we construct a total of $16$ new classes of $2$-to-$1$ mappings, which are not QM-equivalent to any existing $2$-to-$1$ polynomials. Among these, $9$ classes are of the form $cx + {\rm Tr}_{q^l/q}(x^d)$, and $7$ classes have the form $cx + {\rm Tr}_{q^l/q}(x^{d_1} + x^{d_2})$. These new infinite classes explain most of numerical results by MAGMA under the conditions that $q=2^k$, $k>1$, $kl<14$ and $c \in \gf_{q^l}^*$. Finally, we construct some binary linear codes using the newly proposed $2$-to-$1$ mappings of the form $cx + {\rm Tr}_{q^l/q}(x^d)$. The weight distributions of these codes are also determined. Interestingly, our codes are self-orthogonal, minimal, and have few weights.
\end{abstract}


\section{Introduction}
 Let $f$ be a mapping from $\gf_{2^n}$ to $\gf_{2^n}$. $f$ is said to be a $2$-to-$1$ mapping over $\gf_{2^n}$ if and only if for any $a\in\gf_{2^n}$, $a$ has either $0$ or $2$ preimages of $f$ \cite{mesnager2019two}. The $2$-to-$1$ mapping has a wide range of applications in combinatorial mathematics \cite{kolsch2024classifications,cesmelioglu2015bent,carlet2011dillonʼs}, coding \cite{li2021binary,mesnager2014several, berger2003application,delsarte1998association,helleseth2005error},   cryptography \cite{zhang2025balancedbooleanfunctionsfewvalued,yuan2021twoandinvolution,idrisova2019algorithm,qu2025parametric,nyberg1993differentially}, and other fields. For example, any hyperoval in the projective plane  $\mathrm{PG}(2,2^n)$ over $\gf_{2^n}$ can be generated by the so-called o-polynomials \cite{cherowitzo1988hyperovals}, which are polynomials $f$ satisfying that $f(x)+\beta x$ is $2$-to-$1$ for any $\beta\in\gf_{2^n}^{*}$ \cite{carlet2011dillonʼs}. 
 In 2019, Mesnager and Qu \cite{mesnager2019two}  provided the first systematic study of $2$-to-$1$ mappings over finite fields. Since then, many constructions and applications of $2$-to-$1$ mappings have been proposed, e.g. \cite{li2021furtherstudy,mesnager2023several,yuan2021twoandinvolution}. Note that the definition of $2$-to-$1$ mappings is suitable for finite fields with any characteristic, and has been generalized into $m$-to-$1$ mappings for any positive integer $m$, e.g. \cite{niu2023characterizations,zheng2024many}. This paper only considers $2$-to-$1$ mappings over finite fields with even characteristic.   In the following,  we first summarize the existing $2$-to-$1$ mappings over $\gf_{2^n}$ from polynomials with few terms, low-degree polynomials, polynomials of the form  $(x^{2^k}+x+\delta)^s+cx$ and polynomials of the form   $(x^{2^k}+x+\delta)^{s_1}+(x^{2^k}+x+\delta)^{s_2}+cx$.
 
\subsection{Polynomials with few terms}
Firstly, it is trivial that there is no $2$-to-$1$ monomial over $\gf_{2^n}$. For $2$-to-$1$ binomials, in 2019, Mesnager and Qu \cite{mesnager2019two} obtained some infinite classes from the so-called o-monomials; see Table \ref{1-1}. Furthermore, recently,   K\"olsch and Kyureghyan \cite{kolsch2024classifications} showed that the classification of $2$-to-$1$ binomials over $\gf_{2^n}$ is equivalent to that of o-monomials, which is a well-studied and elusive problem in finite geometry. That is, all $2$-to-$1$ binomials over $\gf_{2^n}$ are induced by o-monomials. 

\begin{table}[H]
	\centering
	\caption{$2$-to-$1$ polynomials of the form $x^k+x$ over $\gf_{2^n}$}
	\label{1-1}
	\begin{tabular}{cc}
		\toprule[1pt]
		 $n$& $k$\\
		\midrule
		 \multirow{4}*{positive odd}& 2\\
		 ~& 6\\
		 ~& $2^{\delta}+2^\pi$, where $\delta=\frac{n+1}{2}$, $4\pi\equiv 1\pmod{n}$\\
		 ~& $3 \cdot 2^{\delta}+4$, where $\delta=\frac{n+1}{2}$\\
         \midrule
		 $n\leq 7$, even& $2^n-2$\\
		\bottomrule
	\end{tabular}
\end{table}

Regarding polynomials with more terms, Li et al. \cite{li2021furtherstudy} obtained some $2$-to-$1$ trinomials and quadrinomials, see Tables \ref{1-3} and \ref{1-4}, respectively.  It is worth mentioning that these classes in Table \ref{1-3} comprehensively account for all instances of $2$-to-$1$ trinomials of the form $x^k+\beta x^l+\alpha x \in \gf_{2^n}[x]$ where $n \leq 7$.

In 2022, Mesnager et al. \cite{mesnager2022more} studied $2$-to-$1$ mappings without fixed points and exhibited the involutions derived from them, where six classes of the form $x^k+x^l+x^s+x^3+x^2+x$ over $\gf_{2^n}$ were presented, and a conjecture was offered about a class of $2$-to-$1$ quadrinomials. Furthermore, in 2024, Kölsch and Kyureghyan \cite{kolsch2024classifications} proved this conjecture in a more general form (hexanomials) by showing that it is induced by the Glynn I o-monomial. All $2$-to-$1$ hexanomials are summarized in Table \ref{1-6}.

\begin{table}[H]
	\centering
	\caption{$2$-to-$1$ trinomials of the form $x^k+\beta x^l+\alpha x$ over $\gf_{2^n}$}
	\label{1-3}
	\begin{tabular}{ccccc}
		\toprule
		 $k$& $\beta$& $l$& $\alpha$& $n$\\
		\midrule
		 $2^n-2^m$& 1& $2^n-2^m-1$& $\alpha^{2^m}+\alpha+1=0$& $2m$\\
		 $\frac{2^{n-1}+2^m-1}{3}$& 1& $2^m$& $\gf_{2^2}\verb|\|\gf_{2}$& $2m$, $m$ odd\\
		\bottomrule
	\end{tabular}
\end{table}

\begin{table}[H]
	\centering
	\caption{$2$-to-$1$ quadrinomials of the form $x^k+x^l+x^d+x^s$ over $\gf_{2^n}$}
	\label{1-4}
	\begin{tabular}{ccccc}
		\toprule
		$s$& $k$& $l$& $d$& $n$\\
		\midrule
		\multirow{12}*{1}& $2^{m+1}+2$& $2^{m+1}$& $2$& \multirow{10}{*}{$2m+1$}\\
		~& $2^{m+1}+2$& $2^{m+1}+1$& $2$& ~\\
		~& $2^{m+1}+4$& $2^{m+1}+2$& $2$& ~\\
		~& $2^n-2^{m+1}+2$& $2^{m+1}$& $2$& ~\\
		~& $2^{n}-2$& $2^n-2^{m+1}$& $2^n-2^{m+1}-2$& ~\\
		~& $2^{n}-2$& $2^n-2^{m+1}$& $2^{m+1}-1$& ~\\
		~& $2^{n}-2$& $2^{n-1}+1$& $2^{n-1}-2$& ~\\
		~& $2^{n}-2$& $2^n-4$& 3& ~\\
		~& 6& 4& 3& ~\\
		~& 6& 5& 3& ~\\
		\cmidrule{2-5}
		~& $2^{2m}+2^m$& $2^{2m}+1$& $2^m+1$& $3m$\\
		~& $2^{2m}+1$& $2^{m+1}$& $2^m+1$& $3m$, $m \not \equiv 1 \pmod{3}$\\
		\midrule
		3& $2^{2m+1}+1$& $2^{m+1}+1$& 4& $3m$, $m$ odd\\
		\bottomrule
	\end{tabular}
\end{table}

\begin{table}[H]
	\centering
	\caption{$2$-to-$1$ hexanomials over $\gf_{2^{n}}$}
	\label{1-6}
	\begin{tabular}{ccc}
		\toprule
		$n$& function& condition\\
		\midrule
		\multirow{6}{*}{$n=2m+1$}& $x^{16}+x^{12}+x^8+x^6+x^4+x^3$& $m \not \equiv 1 \pmod{3}$\\ 
		\cmidrule{2-3}
		~& $x^{2^{m+1}+2^m}+x^{2^{m+1}}+x^{2^m}+x^3+x^2+x$& \multirow{4}{*}{$m\in \mathbb{N}$}\\
		~& $x^{2^{m+1}+2}+x^{2^{m+1}+1}+x^{2^{m+2}}+x^3+x^2+x$& ~\\
		~& $x^{2^m+2}+x^{2^m+1}+x^{2^m}+x^3+x^2+x$& ~\\
		~& $x^{2^{m+2}+2}+x^{2^{m+2}+1}+x^4+x^3+x^2+x$& ~\\
		\cmidrule{2-3}
		~& $x^{2^{m+1}+2}+x^{2^{m+1}}+ax^3+(a+1)x^2+ax$& $a \in \gf_{2^n}^*$\\
		\midrule
		$n=2m$& $x^{2^m+2}+x^{2^{m+1}+1}+x^{2^{m+1}+2^m}+x^3+x^2+x$& $m$ even\\
		\bottomrule
	\end{tabular}
\end{table}
\subsection{Low-degree polynomials (degree $\leq 5$)}
Classifications of $2$-to-$1$ polynomials with low degrees are interesting. 
Let $f(x)=\sum_{i=0}^{d}a_ix^i\in \gf_{2^n}[x]$. It is clear that $f$ is a $2$-to-$1$ mapping if and only if so is $f_1(x) = bf(x+c)+d$, where $b,c,d\in\gf_q$ with $b\neq0$. Hence, W.L.O.G., we consider $f$ with normalized form, i.e., $f(x)$ is monic, $f(0)=0$, and when $d$ is odd, the coefficient of $x^{d-1}$ is $0$. In 2019, Mesnager and Qu \cite{mesnager2019two} determined all $2$-to-$1$ polynomials of degree $\leq 4$ over $\gf_{2^n}$, see Table \ref{1-2}, where $\tr_{2^n}(\cdot)$ denotes the absolute trace function over $\gf_{2^n}$.
\begin{table}[H]
	\centering
	\caption{$2$-to-$1$ low-degree polynomials of the form $f(x)=\sum_{i=0}^{d}a_i x^i$ over $\gf_{2^n}$}
	\label{1-2}
	\begin{tabular}{cc}
		\toprule
		function& condition\\
		\midrule
        $x^2+ax$ & $a\neq0$ \\
		 $x^4+a_2 x^2$& $a_2\ne 0$ \\
		 $x^4+a_2x^2+a_1x$& $a_1\ne 0$, $\tr_{2^n}\left(\frac{a_2^3}{a_1^2}\right)\ne \tr_{2^n}(1)$ \\
		 $x^4+a_3x^3+a_2x^2+a_1x$& $a_3 \ne 0$, $a_2^2=a_1a_3$, $m$ odd\\
		\bottomrule
	\end{tabular}
\end{table}

Furthermore, utilizing the well-known Hasse-Weil bound, Li et al. \cite{li2021furtherstudy}  completely determined $2$-to-$1$ polynomials with degree $5$ by proving that when $n>9$, $f(x)=x^5+a_3x^3+a_2x^2+a_1x$ is not $2$-to-$1$ over $\gf_{2^n}$, where $a_3, a_2, a_1 \in \gf_{2^n}$. As for $3 \leq n\leq 9$, they obtained some $2$-to-$1$ mappings by MAGMA of the form $f(x)=x^5+a_3x^3+a_2x^2+a_1x$ over $\gf_{2^n}$, see Table \ref{n=3}, where $\gamma$ is a primitive element of $\gf_{2^3}$.

	\begin{table}[H]
		\caption{$\left(a_3,a_2,a_1\right)\in\gf_{2^3}^{3}$ such that $f(x)$ is $2$-to-$1$} 		\label{n=3}
		\centering
		\begin{tabular}{c c|c c|c c}	
			\toprule
			No. & $\left(a_3,a_2,a_1\right)$ & No.  &  $\left(a_3,a_2,a_1\right)$ & No. & $\left(a_3,a_2,a_1\right)$ \\
			\midrule
			$1$ & $\left(1,\gamma,\gamma^5\right)$ & $2$ & $\left(1,\gamma^2,\gamma^3\right)$ & $3$ & $\left(1,\gamma^4,\gamma^6\right)$ \\
			$4$ & $\left(\gamma,1,\gamma^5\right)$ & $5$ & $\left(\gamma,\gamma^2,\gamma\right)$ & $6$ & $\left(\gamma,\gamma^6,1\right)$ \\
			$7$ & $\left(\gamma^2,1,\gamma^3\right)$ & $8$ & $\left(\gamma^2,\gamma^4,\gamma^2\right)$ & $9$ & $\left(\gamma^2,\gamma^5,1\right)$ \\
			$10$ & $\left(\gamma^3,\gamma^2,\gamma^4\right)$ & $11$ & $\left(\gamma^3,\gamma^3,\gamma^2\right)$ & $12$ & 
			$\left(\gamma^3,\gamma^5,\gamma^5\right)$ \\
			$13$ & $\left(\gamma^4,1,\gamma^6\right)$ & $14$ & $\left(\gamma^4,\gamma,\gamma^4\right)$ & $15$ & 
			$\left(\gamma^4,\gamma^3,1\right)$ \\
			$16$ & $\left(\gamma^5,\gamma,\gamma^2\right)$ & $17$ & $\left(\gamma^5,\gamma^5,\gamma\right)$ & $18$ & $\left(\gamma^5,\gamma^5,\gamma^6\right)$ \\
			$19$ & $\left(\gamma^6,\gamma^3,\gamma^3\right)$ & $20$ & $\left(\gamma^6,\gamma^4,\gamma\right)$ & $21$ &
			$\left(\gamma^6,\gamma^6,\gamma^4\right)$\\
			$22$ & $\left(a_3,0,0\right)$, $a_3\in\gf_{2^3}^{*}$ & $23$ & $\left(0,a_2,0\right)$, $a_2\in\gf_{2^3}^{*}$  & & \\
			\bottomrule
		\end{tabular}
	\end{table}

\subsection{Polynomials of the Kloosterman form}

In 2021, Yuan \cite{yuan2021twoandinvolution} investigated the relationship between $2$-to-$1$ mappings and involutions. Using AGW-like criterion,   they constructed $8$ classes of $2$-to-$1$ polynomials of the Kloosterman form $(x^{2^k}+x+\delta)^s+cx$ in $\gf_{2^{2m}}$, see Table \ref{1-5}, where $\tr_{2^{2m}/2^m}(\cdot)$ denotes the relative trace function from $\gf_{2^{2m}}$ to $\gf_{2^m}$.  Furthermore, in \cite{mesnager2022more}, Yuan et al. proposed $3$ new classes of $2$-to-$1$ mappings of the form $(x^{2^k}+x+\delta)^{s_1}+(x^{2^k}+x+\delta)^{s_2}+cx$ over $\gf_{2^{2m}}$ and derived their involutions without fix points, see Table \ref{1-7}.

\begin{table}[H]
	\centering
	\caption{$2$-to-$1$ polynomials of the form $(x^{2^k}+x+\delta)^s+cx$ over $\gf_{2^{2m}}$}
	\label{1-5}
	\begin{tabular}{cccc}
		\toprule
		$k$& $m$& $s$& condition\\
		\midrule
		\multirow{4}*{1}& \multirow{4}*{even}& $2^m+1$& \multirow{4}*{$\tr_{2^{2m}}(\delta)=1,c=1$}\\
		~& ~& $2^{2m-2}+2^{m-2}$& ~\\
		~& ~& $2^{2m-1}+2^{m-1}$& ~\\
		~& ~& $\frac{2^{2m}+2^m+1}{3}$& ~\\
		\midrule
		\multirow{3}*{$m$}& $\mathrm{gcd}(m,i)=1$& $2^i+1$& $\tr_{2^{2m}/2^m}(\delta^2+c^{2^{m-i}}\delta)\ne 0$\\
		~& $m \in \mathbb{N}$& $2^m+2^i+1$& $\tr_{2^{2m}/2^m}(\delta)^{2^i+2}+c\tr_{2^{2m}/2^m}(\delta)\ne 0$\\
		~& $ m \in \mathbb{N}$& $2^{2m-2}+2^m-2^{m-2}$& $\delta \in \gf_{2^{2m}}\verb|\|\gf_{2^m}$, $c \in \gf_{2^m}^*$, $\tr_{2^m}(c^{-1}+1)=1$\\
		\midrule
		$2m$& $ m \in \mathbb{N}$& $(2^{2m-1}-2^{m-1}+1)(2^{2m}-1)+1$& $ \delta \in \gf_{2^{4m}}\verb|\|\gf_{2^{2m}}$, $c \in \gf_{2^{2m}}\verb|\|\gf_{2^m}$\\
		\bottomrule
	\end{tabular}
\end{table}


\begin{table}[H]
	\centering
	\caption{$2$-to-$1$ polynomials of the form $(x^{2^k}+x+\delta)^{s_1}+(x^{2^k}+x+\delta)^{s_2}+cx$ over $\gf_{2^{2m}}$}
	\label{1-7}
	\begin{tabular}{cccccc}
		\toprule
		$k$& $s_1$& $s_2$& $c$& $m$& condition\\
		\midrule
		2& $2^m+1$& $2^m+1$& 1& even& $\tr_{2^{2m}}(\delta)=1,\delta \in \gf_{2^{2m}}$\\
		2& $\frac{2^{2m}+2^m+1}{3}$& $\frac{2^{2m+1}-2^m-1}{3}$& 1& even& $\tr_{2^{2m}}(\delta)=1,\delta \in \gf_{2^{2m}}$\\
		$2^m$& $2^m+2$& $2^m+3$& $c \in \gf_{2^m}^*$& $m\in \mathbb{N}$& \makecell[c]{$ \tr_{2^{2m}/2^m}(\delta)^3+\tr_{2^{2m}/2^m}(\delta)^2+c \ne 0 $\\$ \delta \in \gf_{2^{2m}},\tr_{2^{2m}/2^m}(\delta) \notin \gf_2$}\\
		\bottomrule
	\end{tabular}
\end{table}

Now we recall the switching method. The fundamental idea behind the switching method is to explore whether a small modification of a function $G(x)$ with certain cryptographic properties can preserve those properties. 
Suppose $G(x)$ and $R(x)$ are polynomials  over the finite field $\gf_{q^l}$ and let $F(x)=G(x)+\tr_{q^l/q}(R(x))$. The function $F(x)$ can be viewed as a relatively ``small" modification of $G(x)$. This method was first introduced by Dillon \cite{dillon2006apn} to study APN functions and was further developed by Edel and Pott \cite{edel2008new}. In 2009, Charpin and Kyureghyan \cite{charpin2009does} introduced the switching method into the study of permutation polynomials. They proposed that when $G(x)$ is a permutation over $\gf_{p^n}$, where $p$ is a prime. In 2014, Qu et al. \cite{qu2013constructing} extended this method to the study of $4$-uniform permutations. In 2022, Li et al.\cite{li2022dillon} generalized Dillon's switching method to characterize the exact $c$-differential uniformity functions. In summary, the switching method is compelling in constructing cryptographic functions, e.g., functions with a low ($c$-)differential uniformity  \cite{xu2016constructing,wu2021new,qu2013constructing,chen2017equivalent,peng2017new,li2022dillon}, permutation polynomials \cite{kyureghyan2011constructing, ShaJiang2025New,charpin2009does, li2018permutationandtrinomial, Li2017NewConstructions,li2019compositional,akbary2011constructing,rai2025permutation}, planar functions \cite{pott2010switching}, bent functions \cite{mesnager2014several, xie2022new, xu2016constructing}. Therefore, utilizing this method to investigate $2$-to-$1$ mappings is a natural extension. 

In the process of delving into the exploration of $2$-to-$1$ polynomials, we notice that Charpin \cite{charpin2009does} obtained several results about constructing $2$-to-$1$ polynomials of the form $cx+\tr_{p^n}(H(x))$. To enrich the construction outcomes of $2$-to-$1$ mappings, this article introduces the generalized switching method\footnote{Note that in this paper, we only consider the case where $G$ is a monomial, which is not $2$-to-$1$, and thus we call the method the generalized switching method.} to construct $2$-to-$1$ mappings of the form $cx+\tr_{q^l/q}(H(x))$ over $\gf_{q^k}$, where $q=2^k$. In this paper, we first obtain some numerical results, including all $2$-to-$1$ mappings of the form $cx^{d_1}+\tr_{q^l/q}(x^{d_2})$ and  $cx^d+\tr_{q^l/q}(x^{d_1}+x^{d_2})$, where $q=2^k$, $k>1$, $kl<14$ and $c \in \gf_{q^l}^*$. On this basis, $16$ new infinite classes of $2$-to-$1$ polynomials, including $9$ classes of the form $cx+\tr_{q^l/q}(x^d)$ and $7$ classes of the form $cx+\tr_{q^l/q}(x^{d_1}+x^{d_2})$, are proved. These constructed $2$-to-$1$ polynomials account for the majority of numerical results. We mainly use two methods to prove these infinite classes. Taking the first construction form as an example, we elaborate on the detailed procedures of these two methods. When the Hamming weight $d$ is small, we can use the first method, which is called the elementary approach. For any $a \in \gf_{q^l}$, let $f(x)=cx+\tr_{q^l/q}(x^d)=a$ and $u=cx+a$. Then $u=\tr_{q^l/q}(x^d)\in \gf_{q}$ and $x=c^{-1}(u+a)$. Plugging $x=c^{-1}(u+a)$ into $f(x)=a$, we transform the equations of high degree over $\gf_{q^l}$ into the equations of low degree over $\gf_{q}$. Consequently, we use the characterization conditions of low-degree equations over finite fields to demonstrate that such low-degree equations possess zero or exactly two solutions in $\gf_{q}$. In the second method, we prove that for any $a\in \gf_{q^l}$, let $f(x+a)+f(a)=0$, there are $2$ solutions in $\gf_{q^l}$. By expanding the equation $f(x+a) + f(a) = 0$, we obtain $cx=\tr_{q^l/q}((x+a)^d+a^d)\in \gf_{q}$, and the original equation can be transformed into the form of a Dickson polynomial. According to the theory of Dickson polynomials, it is deduced that the original equation has two solutions. In the process of some proofs, the well-known Dobbertin's multivariate method \cite{dobbertin2002uniformly} is used. 
Finally, to show that our newly $2$-to-$1$ mappings are useful, we applied some newly $2$-to-$1$ mappings of the form $cx+\tr_{q^l/q}(x^d)$ in this paper to construct linear codes with few weights. Moreover, the weight distributions of these codes are also determined. In addition, these newly linear codes are self-orthogonal, minimal, and with few weights, which have nice applications in many areas, including quantum codes \cite{calderbank1997quantum}, lattices \cite{wan1998characteristic}, secret sharing \cite{massey1993minimal}, etc. 

The rest of this paper is organized as follows. Section \ref{preliminaries} introduces preknowledge and some useful lemmas. Sections  \ref{monomial} and \ref{binomial}  contain $2$-to-$1$ polynomials of the forms $cx+\tr_{q^l/q}(x^d)$  and   $cx+\tr_{q^l/q}(x^{d_1}+x^{d_2})$, respectively. Furthermore, in Section \ref{linear}, we give some results of applying $2$-to-$1$ mappings of the form $cx+\tr_{q^l/q}(x^d)$ to construct linear codes. Finally, Section \ref{conclusion} concludes the work of this paper. 

\section{Preliminaries}\label{preliminaries}
In this section, we introduce the basic notions about this paper, including finite fields, $2$-to-$1$ mappings, and linear codes.
\subsection{Finite fields}
$\mathbb{F}_{2^n}$ denotes the finite field with $2^n$ elements. $\gf_{2^n}^*$ denotes all nonzero elements in $\gf_{2^n}$. The relative trace function from $\gf_{q^l}$ to $\gf_{q}$ is defined as $\tr_{q^l/q}(\alpha)=\alpha+\alpha^q+\cdots+\alpha^{q^{l-1}}$ for any $a\in\gf_{q^l}$. When $q$ is prime, it becomes an absolute trace function and is denoted by $\tr_{q^l}(\alpha)$ for simplicity. The trace function plays a fundamental role in finite field theory. Its basic properties are as follows.
\begin{Lemma}\cite{lidl1997finite}
	\label{trace}
	The trace function $\tr_{q^l/q}(\cdot)$ from $\gf_{q^l}$ to $\gf_{q}$ satisfies the following properties:
	\begin{enumerate}[(1)]
		\item $\tr_{q^l/q}(\alpha+\beta)=\tr_{q^l/q}(\alpha)+\tr_{q^l/q}(\beta)$ for all $\alpha,\beta\in \gf_{q^l}$;
		\item  $\tr_{q^l/q}(c\alpha)=c\tr_{q^l/q}(\alpha)$, for all $c\in \gf_{q}, \alpha\in \gf_{q^l}$;
		\item $\tr_{q^l/q}(\cdot)$ is a linear transformation from $\gf_{q^l}$ to $\gf_{q}$, where both $\gf_{q^l}$ and $\gf_{q}$ are viewed as vector spaces over $\gf_{q}$;
		\item $\tr_{q^l/q}(a)=la$ for all $a \in \gf_{q}$;
		\item  $\tr_{q^l/q}(\alpha^q)=\tr_{q^l/q}(\alpha)$ for all $\alpha \in \gf_{q^l}$.
	\end{enumerate}
\end{Lemma}
Since a finite field is a finite set, any mapping or function over finite fields could be expressed as a polynomial using Lagrange interpolation. In this paper, unless otherwise specified, mappings, functions, and polynomials can be regarded as the same research object. The study of important functions on finite fields, such as APN functions, permutation polynomials, and $2$-to-$1$ mappings, ultimately boils down to determining the number of solutions to specific equations. The following four lemmas characterize the number of solutions to low-degree equations on finite fields.
\begin{Lemma}\cite{lidl1997finite}
	Let $u,v \in \gf_{2^n}$, $u \ne 0$. The quadratic equation $x^2+ux+v=0$ has solutions in $\gf_{2^n}$ if and only if $\tr_{2^n}(\frac{v}{u^2})=0$.
\end{Lemma}
The cubic polynomial $f$ over the finite field $\gf_{2^n}$ is denoted as $f=(1,1,1)$ if $f$ can be decomposed into the product of three affine polynomials. If it is decomposed into the product of an affine polynomial and a quadratic irreducible polynomial, it is $f=(1,2)$. If $f$ is irreducible, then $f = (3)$.
\begin{Lemma}\label{cubic}\cite{berlekamp1967solution, williams1975note}
	Let $a,b \in \gf_{2^n}^*$, $f(x)=x^3+ax+b$, $h(y)=y^2+by+a^3$, $y_1,y_2$ are 2 roots of $h(y)$. Then
	\begin{enumerate}[(i)]
		\item $f=(1,1,1)$ if and only if $\tr_{2^n}(\frac{a^3}{b^2})=\tr_{2^n}(1)$, when n is even (odd), $y_1,y_2$ are cubic elements in $\gf_{2^n}$($\gf_{2^{2n}}$);
		\item $f=(1,2)$ if and only if $\tr_{2^n}(\frac{a^3}{b^2})\ne\tr_{2^n}(1)$;
		\item $f=(3)$ if and only if $\tr_{2^n}(\frac{a^3}{b^2})=\tr_{2^n}(1)$, and when n is even (odd), $y_1,y_2$ are not cubic elements in $\gf_{2^n}$($\gf_{2^{2n}}$).
	\end{enumerate}
\end{Lemma}
The following lemma introduces a method for finding solutions of cubic equations.
\begin{Lemma}\cite{williams1975note}\label{cubics soulution}
	Let $f(x)=x^3+ax+b \in \gf_{2^n}[x]$ and $b\ne 0$. If $y_1$ is a solution to the equation $y^2+by+a^3=0$ and $\epsilon$ is a solution to the equation $x^3=y_1$, then $r=\epsilon+\frac{a}{\epsilon}$ is a solution to the equation $f(x)=0$.
\end{Lemma}
The following lemma introduces the characterization of the factorization of quartic polynomials over the finite field $\gf_{2^n}$.
\begin{Lemma}\label{quartic}\cite{leonard1972quartics}
	Let $f(x)=x^4+a_2x^2+a_1x+a_0$, $a_i\in \gf_{2^n}$, $a_0a_1\ne 0$ and $f_1(y)=y^3+a_2y+a_1$. The solutions of $f_1(y)=0$ in $\gf_{2^n}$ are denoted as $r_1,r_2,r_3$(if they exist) and we use the notation $w_i=a_0\frac{r_i^2}{a_1^2}$, then
	\begin{enumerate}[(i)]
		\item $f=(1,1,1,1)$ if and only if $f_1=(1,1,1)$ and  $\tr_{2^n}(w_1)=\tr_{2^n}(w_2)=\tr_{2^n}(w_3)=0$;
		\item $f=(2,2)$ if and only if $f_1=(1,1,1)$, $\tr_{2^n}(w_1)=0$ and $\tr_{2^n}(w_2)=\tr_{2^n}(w_3)=1$;
		\item $f=(1,3)$ if and only if $f_1=(3)$;
		\item $f=(1,1,2)$ if and only if $f_1(1,2)$ and $ \tr_{2^n}(w_1)=0$;
		\item $f=(4)$ if and only if $f_1=(1,2)$ and $ \tr_{2^n}(w_1)=1$.
	\end{enumerate}
\end{Lemma} 
Dickson polynomials play a crucial role in studying the solutions of polynomials over finite fields. The Dickson polynomials of the first kind of degree $r$ in $\gf_{q}$ is defined as $$D_r(x,a)=\sum_{i=0}^{\lfloor\frac{r}{2}\rfloor}\frac{r}{r-i}\binom{r-i}{i}(-a)^ix^{r-2i}.$$
\begin{Lemma}\label{gcd}\cite{lidl1993theory}
	$D_r(x,a)$ is permutation in $\gf_{q}$ if and only if $\mathrm{gcd}(r,q^2-1)=1$.
\end{Lemma}

\subsection{$2$-to-$1$ mappings over finite fields}
In 2019, the definition of $2$-to-$1$ mappings was given by Mesnager and Qu \cite{mesnager2019two}.
\begin{Def}
	Let $f$ be a mapping from $\gf_{2^n}$ to itself. $f$ is said to be a $2$-to-$1$ mapping if and only if for all $a\in \gf_{2^n}$, $\#f^{-1}(a)\in \{0,2\}$.		
\end{Def}

According to the above definition, the following lemma is clear.

\begin{Lemma}\label{$2$-to-$1$} 
	Let $f$ be a function over $\gf_{2^n}$. Then the following statements are equivalent.
    \begin{enumerate}[(1)]
        \item $f$ is $2$-to-$1$ over $\gf_{2^n}$;
        \item for any $a\in\gf_{2^n}$, the equation $f(x)=a$ has 0 or 2 solutions in $\gf_{2^n}$;
        \item for any $a\in \gf_{2^n}$, equation $f(x+a)+f(a)=0$ only has 2 solutions in $\gf_{2^n}$.
    \end{enumerate}
\end{Lemma}

Permutation polynomials play a crucial role in the construction of $2$-to-$1$ polynomials. Below is a description of the method for constructing $2$-to-$1$ mappings over finite fields through permutations.
\begin{Lemma} \cite{mesnager2019two}
	Let $G:\gf_{2^n}\rightarrow\gf_{2^n}$ be a permutation, $H:\gf_{2^n}\rightarrow\gf_{2^n}$ be $2$-to-$1$. Then $F_1(x)=H(G(x))$ and $F_2(x)=G(H(x))$ are $2$-to-$1$ in $\gf_{2^n}$.
\end{Lemma}

To determine whether a monomial is a permutation polynomial, the following lemma can be applied.

\begin{Lemma}
\cite{lidl1997finite}
Let $d$ be a positive integer. Then $x^d$ is a permutation in $\gf_{q}$ if and only if $\mathrm{gcd}(d,q-1)=1$.
\end{Lemma}

Therefore, we have the following definition.
\begin{Def}\label{QM}
	Let $ f $ and $ g $ be two polynomials over the finite field $ \mathbb{F}_q $, satisfying $ f(x) = a g(bx^d) $, where $ a, b \in \mathbb{F}_q^\ast $ and $ 1 \le d \le q-1 $ is an integer coprime to $ q-1 $. Then $f$ and $g$ are said to be quasi-multiplicative (QM for short) equivalent, and $f$ is $2$-to-$1$ if and only if $g$ is also $2$-to-$1$.
\end{Def}

\subsection{Cyclotomic equivalence}
When conducting computational searches for $2$-to-$1$ mappings using MAGMA, careful consideration of mapping equivalence relations is essential. To optimize the search process, we systematically exclude equivalent function classes, thereby significantly reducing the search space. As a concrete illustration, we focus on $2$-to-$1$ mappings of the form $f(x) = cx + \tr_{q^l/q}(x^d)$ over the finite field $\gf_{q^l}$. This approach leverages fundamental properties of the trace function, particularly its invariance under Frobenius automorphism: for any $a \in \gf_{q^l}$, $\tr_{q^l/q}(a^q) = \tr_{q^l/q}(a)$. Before proceeding with the search algorithm, we first formalize the notion of cyclotomic equivalence as follows:

\begin{Def}
	Let $d_1$ and $d_2$ be two positive integers less than $q^l-1$. If there exists a positive integer $i$ such that $d_1\equiv d_2\cdot q^i \pmod {q^l-1}$, then $d_1$, $d_2$ are said to be cyclotomically equivalent.
\end{Def}

Together with the above definition and the QM-equivalence, we obtain some experimental results about $2$-to-$1$ mappings of the form $cx^{d_1}+\tr_{q^l/q}(x^{d_2})$ and $cx^d+\tr_{q^l/q}(x^{d_1}+x^{d_2})$, see Tables \ref{3-1} and \ref{4-1}, respectively.

\subsection{Walsh transform}

Let $f$ be a mapping from $\gf_{2^n}$ to itself with $f(0)=0$. The Walsh transform of $f$ at $(a,b)\in \gf_{2^{n}}\times\gf_{2^{n}}$ is defined as  
\begin{eqnarray}\label{walsh}
 W_f(a,b)=\sum_{x \in \gf_{2^{n}}}(-1)^{\tr_{2^n}(ax+bf(x))}.
 \end{eqnarray}
The multiset $\{W_f(a,b):a,b \in \gf_{2^n}\}$ is called the Walsh spectrum of $f$.  

\begin{Lemma}\label{qua}\cite{li2021binary}
Let $Q(x)$ be a quadratic function over $\gf_{2^n}$ and $\varphi_{a,b}(x)=\tr_{2^n}(ax+bQ(x))$. The bilinear form $B_{\varphi_{a,b}}(x,y)$ of $\varphi_{a,b}(x)$ is given by $$B_{\varphi_{a,b}}(x,y)=\varphi_{a,b}(x+y)+\varphi_{a,b}(x)+\varphi_{a,b}(y).$$
The kernel $V_{\varphi_{a,b}}$ of $B_{\varphi_{a,b}}$ is defined as $$V_{\varphi_{a,b}}=\{y \in \gf_{2^n}|B_{\varphi_{a,b}}(x,y)=0, \forall x \in \gf_{2^n}\}.$$ The rank of $\varphi_{a,b}$ is $\rank(\varphi_{a,b})=n-d$, where $d = \dim_{\gf_2}(V_{\varphi_{a,b}})$ is the dimension of $V_{\varphi_{a,b}}$ as an $\gf_2$-vector space.
The Walsh transform of $Q(x)$ at $(a,b)\in \gf_{2^n} \times \gf_{2^n}$ is defined as $$W_Q(a,b)=\sum_{x \in \gf_{2^{n}}}(-1)^{\tr_{2^n}(ax+bQ(x))}=
\begin{cases}
\pm2^{\frac{n+d}{2}}, & \text{if $\varphi_{a,b}$ vanishes on $V_{\varphi_{a,b}}$}, \\
0, & \text{otherwise}.
\end{cases}$$
\end{Lemma}

\subsection{Binary linear codes from functions}

In this subsection, we first recall some basic knowledge about binary linear codes, which can be found in \cite{huffman2010fundamentals}. A non-empty subset $C$  of $\gf_2^n$ is called a binary $[n,k]$-linear code if $C$ is a subspace with dimension $k$. Any element of $C$ is called a codeword. For any codeword $\bu\in C$, its (Hamming) weight $\wt(\bu)$ is the cardinality of its support, i.e., $\supp(\bu) = \{ i: 1 \le i\le n ~|~ u_i =1 \}$. For any binary $[n,k]$-linear code $C$, the minimum (Hamming) distance of $C$ is $d=\min_{\bu\neq0}\{\wt(\bu): \bu\in C\}$. If the minimum distance $d$ of $C$ is known, then $C$ can be called a binary $[n,k,d]$-linear code. Moreover, let $A_i$ denote the number of codewords with weight $i$ in $C$. The weight enumerator of $C$ is defined as $1+A_1z+A_2z^2+\cdots+A_nz^n$. The sequence $(1, A_1, A_2,\ldots, A_n)$ is called the weight distribution of $C$. Moreover, $C$ is said to be a $t$-weight code or a code with $t$ weights if the number of nonzero $A_i$ in the sequence $(1, A_1, A_2,\ldots, A_n)$ is equal to $t$. 

For a binary linear code $C$, its dual code $C^{\perp}$ is defined as: 
 $$C^{\perp} = \left\{  \bv: \bv \in C ~|~ \bu \cdot \bv = 0 ~\text{for all}~ \bu\in C \right\}.$$
$C$ is called  self-orthogonal if $C \subseteq C^{\perp}$. It is well-known that if for any $\bu\in C$, $\wt(\bu)\equiv 0\pmod 4$, then $C$ is self-orthogonal. Moreover, we say a codeword $\bu\in C$ covers a codeword $\bv\in C$ if $\supp(\bu)$ contains $\supp(\bv)$.  A codeword $\bu\in C$ is called minimal if $\bu$ covers only the codewords $\bu$ and $\mathbf{0}$, but no other codewords in $C$. $C$ is said to be minimal if every codeword in $C$ is minimal. Furthermore, Aschikhmin and Barg \cite{ashikhmin2002minimal} proved that if $\frac{w_{\min}}{w_{\max}}> \frac{1}{2}$, where $w_{\min}$ and $w_{\max}$ denote the minimum and maximum nonzero weights in $C$, respectively, then $C$ is minimal.

We now recall a general approach to constructing linear codes from functions, see \cite{carlet2005linear}.  Given a function $f: \gf_{2^{n}}\rightarrow \gf_{2^{n}}$, a binary linear code $C_f$ can be defined as  
\begin{eqnarray}\label{cf}
C_f=\{c_{a,b}=(\tr_{2^n}(ax+bf(x)))_{x\in \gf_{2^{n}}^*}:a, b \in \gf_{2^{n}}\}.  
\end{eqnarray}
The code length is $2^n-1$, and the dimension is at most $2n$. The dimension of the linear code $C_f$ equals $2n-d_{K}$, where $K=\{(a,b)\in \gf_{2^{n}}^2:W_f(a,b)=2^n\}$ and $d_{K}$ is the dimension of $K$. For this construction method, the Hamming weight of a codeword in the linear code is given by  
\begin{eqnarray}
    \label{wt}
wt(c_{a,b})&=&\#\{x \in \gf_{2^n}^*:\tr_{2^n}(ax+bf(x))=1\} \nonumber \\
&=&2^{n-1}-\frac{1}{2}\sum_{x\in \gf_{2^n}}(-1)^{\tr_{2^n}(ax+bf(x))} \nonumber\\
&=&2^{n-1}-\frac{1}{2}W_f(a,b).
\end{eqnarray} 
Thus, the weight distribution of the linear code is entirely determined by the Walsh spectrum of $f$.

When  $W_f(a,b) = v_i $ occurs $ N_i $ times, the corresponding codeword weight is $wt(c_{a,b}) = 2^{n-1} - \frac{v_i}{2}$ with frequency $ \frac{N_i}{2^{d_{K}}} $. This is because the code $ C_f $ is linear, and the value of a codeword depends on the coset of $(a,b)$ in the quotient space $ \gf_{2^n}^2/K$. If $(a,b)$ and $(a',b')$ belong to the same coset, their corresponding codewords are identical, i.e., $ c_{a,b} = c_{a',b'} $. Thus, each coset corresponds to a unique codeword.  

The total number of pairs $(a,b) \in \gf_{2^n}^2 $ is $ 2^{2n} $, and the size of $ K_1 $ is $ 2^{d_{K}} $. Therefore, the total number of cosets is $ \frac{2^{2n}}{2^{d_{K}}} = 2^{2n - d_{K}} $, which equals the number of codewords in the linear code $ C_f $. This can be viewed as a $(2n - d_{K})$-dimensional linear space over $ \gf_2 $.  

When there are $ N_i $ pairs $(a,b)$ such that $ W_f(a,b) = v_i $, each coset contains $ 2^{d_{K}} $ such pairs, leading to a codeword weight frequency of $ \frac{N_i}{2^{d_{K}}} $.  

For all $ a, b \in \gf_{2^n}$, the Walsh transform $ W_f(a,b) $ takes only three non-trivial values (i.e., $ \neq 2^n $), denoted as $ v_1, v_2, v_3 $, then the Walsh value distribution can be determined by the following system of equations:  
$$
\begin{cases}  
X_0 + X_1 + X_2 + X_3 = 2^{2n}, \\  
v_0X_0 + v_1X_1 + v_2X_2 + v_3X_3 = 2^{2n}, \\  
v_0^2X_0 + v_1^2X_1 + v_2^2X_2 + v_3^2X_3 = 2^{3n},  
\end{cases}  
$$
where $ X_i $ denotes the number of occurrences of $ v_i $ ($ i = 0,1,2,3 $) in the Walsh spectrum of $ f $. Note that $ (X_0, v_0) = (2^{d_{K}}, 2^n) $. Consequently, the weight distribution of the linear code $C_f$ can be fully determined.

\section{$2$-to-$1$ mappings of the form $cx+\tr_{q^l/q}(x^d)$}\label{monomial}
In this section, we construct several $2$-to-$1$ mappings of the form $cx^{d_1}+\tr_{q^l/q}(x^{d_2})$ over $\gf_{q^l}$. We first obtain some numerical results by MAGMA, see Table \ref{3-1}. The scope we search is as follows: $q=2^k$, $d_1,d_2 \in[1,q^l-2]$, $kl<14$, $c \in \gf_{q^l}^*$, $k,l \in \mathbb{N}^+$.
\begin{table}[H]
	\centering
	\caption{$2$-to-$1$ polynomials of the form $cx+\tr_{q^l/q}(x^d)$ in $\gf_{q^l}$, where $q=2^k$, $kl<14$}
	\label{3-1}
	\begin{tabular}{cccc }
		\toprule
		$(k,l)$& $d$& $c$& Reference\\
		\midrule
		(2,3)& 5& $\gf_4^*$& Theorem \ref{3.2}(\ref{q+1})\\
		\midrule
		\multirow{5}{*}{(3,3)}& 6& $\gf_8^*$& Theorem 3.3(\ref{x^6})\\
		~& 9& $\gf_8^*$& Theorem \ref{3.2}(\ref{q+1})\\
		~& 18& $\gf_8^*$& Theorem \ref{3.2}(\ref{2q+2})\\
		~& 20& $\gf_8^*$& Theorem \ref{3.3}(\ref{2q+4})\\
		~& 34& $\gf_8^*$& Theorem \ref{3.3}(\ref{4q+2})\\
		\midrule
		\multirow{2}{*}{(2,5)}& 5& $\gf_4^*$& Theorem \ref{3.2}(\ref{q+1})\\
		~& 17& $\gf_4^*$& Theorem \ref{3.2}(\ref{q^2+1})\\
		\midrule
		(3,2)& 20& $c^{q-1}\in \gf_4 \backslash \gf_2$&  Theorem \ref{l=2,2q+4}\\
		\midrule
		\multirow{3}{*}{(5,2)}& 10&  $c^{q-1}\in \gf_4\verb|\|\gf_2$& Theorem \ref{Conj1}\\
		~& 68& $c^{q-1}\in \gf_4\verb|\|\gf_2$& Theorem \ref{l=2,2q+4}\\
		~& 272& $c^{q-1}\in \gf_4\verb|\|\gf_2$& Theorem \ref{conj2}\\
		\bottomrule
	\end{tabular}
\end{table}

To narrow the search space, the following three additional conditions are imposed:
\begin{enumerate}[(1)]
	\item since \cite{charpin2009does} has already investigated the case where $k=1$, to avoid obtaining redundant results, it is assumed that $k>1$. 
	\item $d_2$ is the representative element of all positive integers less than $q^l-1$ under the cyclotomic equivalence, and $\tr_{q^l/q}(x^{d_2}) \ne 0$.
	\item Regarding the range of $d_1$, only cases where $d_1$
	within the range $[1,q^l-2]$ is a divisor of $q^l-1$ are considered based on the following lemma.
\end{enumerate} 
\begin{Lemma}\label{divisor}
	Let $\gcd(d_1,q^l-1)=s$. $f(x)=cx^{d_1}+\tr_{q^l/q}(x^{d_2})$ is $2$-to-$1$ if and only if $g(x)=f(x^t)=x^s+\tr_{q^l/q}(x^{d_2t})$ is $2$-to-$1$ where $d_1t\equiv s \pmod {q^l-1}$, $s$ is a divisor of $q^l-1$.
\end{Lemma}
\begin{proof}
	Let $(d_1,q^l-1)=s$, $(t,q^l-1)=1$. Then $(td_1,q^l-1)=s$, $$f\left(x^t\right)=cx^{td_1}+\tr_{q^l/q}\left(x^{d_2t}\right)=cx^s+\tr_{q^l/q}\left(x^{d_2t}\right).$$ Since $g(x)=f(x^t)$, $f(x)$ is QM-equivalent to $g(x)$. Due to Definition \ref{QM}, $f$ is $2$-to-$1$ if and only if $g$ is $2$-to-$1$.
\end{proof}

Then we obtain all $2$-to-$1$ polynomials over $\gf_{q^l}$ of the form $cx^{d_1}+\tr_{q^l/q}(x^{d_2})$ with $q=2^k$, $d_1,d_2 \in[1,q^l-2]$, $kl<14$, $c \in \gf_{q^l}^*$ and $k,l \in \mathbb{N}^+$ in Table \ref{3-1}. The ``Reference" column in the table indicates that this dataset served as an illustrative example for the theorem cited in the ``Reference". The data illustrated that there does not exist a $2$-to-$1$ mapping of the form $cx^{d_1}+\tr_{q^l/q}(x^{d_2})$ when $d_1 \ne 1$. Therefore, we consider the form $cx+\tr_{q^l/q}(x^d)$. All examples in Table \ref{3-1} have been generalized into infinite classes.

\begin{Th}
    \label{3.2}
	Let $l$ be an odd positive integer, $q=2^k$, $k \in \mathbb{N}$, $c\in\mathbb{F}_q^\ast$. Then
	\begin{enumerate}[1)]
		\item $f(x)=cx+\tr_{q^l/q}(x^{q+1})$;\label{q+1}
		\item $f(x)=cx+\tr_{q^l/q}(x^{2q+2})$, where $k$ is odd;\label{2q+2} 
		\item $f(x)=cx+\tr_{q^l/q}(x^{q^2+1})$ \label{q^2+1}
	\end{enumerate}
	are $2$-to-$1$ over $\mathbb{F}_{q^l}$. 
\end{Th}
\begin{proof}
    Since the proof techniques for these cases are analogous, we present a detailed demonstration for the first case as a representative example. 
    According to the definition of $2$-to-$1$ mappings, we show that for any $a\in \gf_{q^l}$, the equation $f(x)=a,$ i.e.,
	\begin{eqnarray}\label{1}
		cx+\tr_{q^l/q}(x^{q+1})=a
	\end{eqnarray}
	has 0 or 2 solutions in $\gf_{q^l}$.
    
	Let $u=cx+a$. Then $u=\tr_{q^l/q}(x^{q+1})\in \gf_{q}$ and $x=\frac{u+a}{c}$ since $c\ne0$. Moreover, since $u\in \gf_{q}$ and $c\in \gf_q^*$, $$x^{q+1}=\left(\frac{u+a}{c}\right)^{q+1}=\frac{u^2+u(a+a^q)+a^{q+1}}{c^2}.$$
    
	Plugging the above equation into Eq. (\ref{1}), we get 
	$$u=\tr_{q^l/q}\left(\frac{u^2+u(a+a^q)+a^{q+1}}{c^2}\right)=\frac{u^2}{c^2}+\tr_{q^l/q}(a+a^q)\frac{u}{c^2}+\frac{\tr_{q^l/q}(a^{q+1})}{c^2},$$
	where the second equality holds due to Lemma \ref{trace}(4) and $l$ being odd. Note that for any $a\in\gf_{q^l}$, $\tr_{q^l/q}(a+a^q)=0$, thus the above equation becomes
	$$u=\frac{u^2}{c^2}+\frac{\tr_{q^l/q}(a^{q+1})}{c^2}.$$
	That is, 
	\begin{eqnarray}\label{2}
		u^2+c^2u+\tr_{q^l/q}(a^{q+1})=0.
	\end{eqnarray}
    
	In case of $c\ne0$, Eq. (\ref{2}) has 0 or 2 solutions in $\gf_{q}$. Hence $f(x)=a$ has 0 or 2 solutions in $\gf_{q^l}$. That is to say, $f(x)$ is $2$-to-$1$ over $\gf_{q^l}$.
\end{proof}
\begin{Th}
\label{3.3}
	Let $k$ and $l$ be odd positive integers, $q=2^k$ and $c\in \gf_{q}^* $. Then
	\begin{enumerate}[1)]
		\item $f(x)=cx+\tr_{q^l/q}(x^6)$;\label{x^6}
		\item $f(x)=cx+\tr_{q^l/q}(x^{4q+2})$; \label{4q+2}
		\item $f(x)=cx+\tr_{q^l/q}(x^{2q+4})$ \label{2q+4}
	\end{enumerate}
	are $2$-to-$1$ over $\gf_{q^l}$. 
\end{Th}
\begin{proof}
    Since the proof techniques for these cases are analogous, we present a detailed demonstration for the first case as a representative example.
	According to Lemma \ref{$2$-to-$1$}, it suffices to show that for any $a \in \gf_{q^l}$, the equation 
	$$f(x+a)+f(a)=0,$$ 
	i.e., 
	\begin{eqnarray}\label{3}
		cx+\tr_{q^l/q}(x^6+a^2x^4+a^4x^2)=0
	\end{eqnarray} 
	has 2 solutions in $\gf_{q^l}$.
	According to Eq. (\ref{3}), we get $cx=\tr_{q^l/q}(x^6+a^2x^4+a^4x^2)\in\gf_{q}$ and $x\in \gf_{q}$ since $c\in\gf_{q}^*$. 
    
    From Lemma \ref{trace}(4),
	$$x^6+\tr_{q^l/q}(a^2)x^4+\tr_{q^l/q}(a^4)x^2+cx=0,$$
	where $l$ is odd and $x\in \gf_{q}$.
	Thus $x=0$ or 
	\begin{eqnarray}\label{4}
		x^5+\tr_{q^l/q}(a^2)x^3+\tr_{q^l/q}(a^4)x=c.
	\end{eqnarray}
    
	We notice that $$x^5+\tr_{q^l/q}(a^2)x^3+\tr_{q^l/q}(a^4)x=D_5(x,\tr_{q^l/q}(a^2))$$ is a Dickson polynomial of degree 5. Recalling that $k$ is odd, $\mathrm{gcd}(5,q^2-1)=1$, and from Lemma \ref{gcd}, $D_5(x,\tr_{q^l/q}(a^2))$ permutes $\gf_{q}$. That is, Eq. (\ref{4}) only has one solution in $\gf_{q}$, and the solution is not 0 since $c \ne 0$. Therefore, Eq. (\ref{3}) has two solutions in $\gf_{q^l}$ and then $f(x)$ is $2$-to-$1$ over $\gf_{q^l}$.
\end{proof}
\begin{Th}
	\label{l=2,2q+4}
	Let $k$ be an odd positive integer and $q=2^k$, $c \in \gf_{q^2}^*$ satisfies $c^{q-1}=\omega$ with $\omega\in \gf_{2^2} \backslash \gf_2$. Then $f(x)=cx+\tr_{q^2/q}(x^{2q+4})$ is $2$-to-$1$ over $\gf_{q^2}$.
\end{Th}
\begin{proof}
	According to Lemma \ref{$2$-to-$1$}, it suffices to show that for any $a\in \gf_{q^2}$, the equation $$f(x+a)+f(a)=0,$$ i.e.,
	\begin{eqnarray}\label{5}
		cx+\tr_{q^2/q}((x+a)^{2q+4}+a^{2q+4})=0
	\end{eqnarray}
	has two solutions in $\gf_{q^2}$.
	
	Let $y=x^{2q}$, $b=a^{2q}$. Plugging the above two equations into Eq. (\ref{5}), we get \begin{eqnarray}\label{6}
		cx+x^4y+x^4b+a^4y+x^2y^2+x^2b^2+a^2y^2=0.
	\end{eqnarray}
    
	Computing the sum of Eq. (\ref{6}) and its $q$th power, then $y=c^{2(1-q)}x^2$, implying $y=\omega^{-2}x^2$.
	Then from Eq. (\ref{6}), 
    \begin{eqnarray*}
        \left(\frac{\omega^2+1}{\omega^4}\right)x^6+\left(\frac{a^2}{\omega^4}+b\right)x^4+\left(\frac{a^4}{\omega^2}+b^2\right)x^2+cx=0.
    \end{eqnarray*}
    
	According to the properties of cubic primitive roots $\omega^2+\omega+1=0$, the above equation becomes $$x^6+(\frac{a^2}{\omega}+b)x^4+(\frac{a^4}{\omega^2}+b^2)x^2+cx=0.$$
    
    Obviously $x=0$, it only needed to prove for any $a\in \gf_{q^2}$, the equation 
	\begin{equation}
		x^5+\left(\frac{a^2}{\omega^4}+b\right)x^3+\left(\frac{a^4}{\omega^2}+b^2\right)x=c \label{7}
	\end{equation}
	has one solution in $\gf_{q^2}^*$.
    
	We notice that $x^5+\left(\frac{a^2}{\omega}+b\right)x^3+\left(\frac{a^4}{\omega^2}+b^2\right)x=D_5\left(x,\frac{a^2}{\omega}+b\right)$ is a Dickson polynomial of degree 5. Recalling that $k$ is odd, $\mathrm{gcd}(5,q^2-1)=1$, and from Lemma \ref{gcd}, $D_5\left(x,\frac{a^2}{\omega^4}+b\right)$ permutes $\gf_{q}$. That is, Eq. (\ref{7}) only has one solution in $\gf_{q}$, and the solution is not 0 since $c\ne 0$. Therefore, the Eq. (\ref{5}) has two solutions in $\gf_{q^2}$ and $f(x)$ is $2$-to-$1$ over $\gf_{q^2}$.
\end{proof}

\begin{Th}
    \label{Conj1}
	Let $k \equiv 1 \pmod 4$ and $q=2^k$, $c \in \gf_{q^2}^*$ satisfies $c^{q-1}=\omega$ with $\omega \in \gf_{2^2} \backslash \gf_2$. Then $f(x)=cx+\tr_{q^2/q}\left(x^{2^\frac{k+1}{2}+2}\right)$ is $2$-to-$1$ over $\gf_{q^2}$. 
\end{Th} 
\begin{proof}
    According to Lemma \ref{$2$-to-$1$}, it suffices to show that for any $a\in \gf_{q^2}$, the equation $$f(x+a)+f(a)=0,$$ i.e.,
	\begin{eqnarray}\label{3.5}
		cx+\tr_{q^2/q}\left(x^{2^{\frac{k+1}{2}}+2}+a^2x^{2^{\frac{k+1}{2}}}+a^{2^{\frac{k+1}{2}}}x^2\right)=0
	\end{eqnarray}
    has two solutions, i.e., one nonzero solution in $\gf_{q^2}$. In the following, we mainly consider the unique nonzero solution. 
    
    Let $$
    \begin{cases}
        y=x^{2^{\frac{k+1}{2}}}, z=y^{2^{\frac{k+1}{2}}}=x^{2q}, t=z^{2^{\frac{k+1}{2}}}=y^{2q}, t^{2^{\frac{k+1}{2}}}=x^4.\\
        u=a^{2^{\frac{k+1}{2}}}, v=u^{2^{\frac{k+1}{2}}}=a^{2q}, s=v^{2^{\frac{k+1}{2}}}=u^{2q}, s^{2^{\frac{k+1}{2}}}=a^4.
    \end{cases}
    $$
    Then we get $x^{2^{\frac{k+1}{2}}+2}=x^2y$, $x^{2^{\frac{k+1}{2}}}=y$, $a^{2^{\frac{k+1}{2}}}=u$, $(x^2y)^q=zt^{\frac{1}{2}}$, $(a^2y)^q=vt^{\frac{1}{2}}$, $(x^2u)^q=zs^{\frac{1}{2}}$. 
    
    From Eq. (\ref{3.5}), 
    \begin{eqnarray}\label{x1,y1}
        cx+x^2y+zt^{\frac{1}{2}}+a^2y+vt^{\frac{1}{2}}+ux^2+zs^{\frac{1}{2}}=0.
    \end{eqnarray}
    Raising Eq. (\ref{x1,y1}) to its $2^{\frac{k+1}{2}}$-th power, it leads to
    \begin{eqnarray}\label{x2,y2}
        c^{2^{\frac{k+1}{2}}}y+y^2z+tx^2+u^2z+sx^2+vy^2+a^2t=0.
    \end{eqnarray}
    
    According to Eq. (\ref{3.5}), we have $$cx=\tr_{q^2/q}\left(x^{2^{\frac{k+1}{2}}+2}+a^2x^{2^{\frac{k+1}{2}}}+a^{2^{\frac{k+1}{2}}}x^2\right) \in \gf_q,$$ and then $x^q=c^{1-q}x=\omega^2 x$. 
    
    Since $k \equiv 1 \pmod 4$, $z=x^{2q}=\omega x^2$, $t=\omega^{2^{\frac{k+1}{2}}}y^2=\omega^2y^2$. Plugging them into Eq. (\ref{x1,y1}) and Eq. (\ref{x2,y2}), we get 
   \begin{equation}
       \label{con1_eq1}
       cx+\omega x^2y+(a^{2q}\omega+a^2)y+(u^q\omega+u)x^2=0
   \end{equation}
 and 
 \begin{equation}
     \label{con1_eq2}
     c^{2^{\frac{k+1}{2}}}y+x^2y^2+(u^2\omega+u^{2q})x^2+(a^{2q}+a^2\omega^2)y^2=0,
 \end{equation}
 respectively. 

   From Eq. \eqref{con1_eq1}, we have $(\omega x^2 + a^{2q}\omega+a^2) y = cx + (u^q\omega+u)x^2$ and thus we need to discuss $\omega x^2 + a^{2q}\omega+a^2=0,$ i.e., $x=a^q+a\omega$ or not.  If $x=a^q+a\omega$ is a nonzero solution, then $a^q+a\omega\neq0$ and $c+(u^q\omega + u)(a^q+a\omega)=0$ by Eq. \eqref{con1_eq1}. Thus, in the following, we divide the proof into three cases: (I) $a^{2q}\omega+a^2=0$; (II) $a^{2q}\omega+a^2\neq0$ and $c+(u^q\omega + u)(a^q+a\omega)=0$; (III) $a^{2q}\omega+a^2\neq0$ and $c+(u^q\omega + u)(a^q+a\omega)\neq0$.

(I)  When $a^{2q}\omega+a^2=0$, after raising it to its $2^{\frac{k+1}{2}}$-th power, we get $u^q+u\omega^2=0$. Then by Eq. \eqref{con1_eq2}, we have $c^{2^{\frac{k+1}{2}}}y+x^2y^2=0$, i.e., $x^{2^{\frac{k+1}{2}}+2}=c^{2^{\frac{k+1}{2}}}$. Note that $$\gcd\left( 2^{\frac{k+1}{2}}+2, 2^{2k}-1 \right)=1$$ since 
  $$\gcd\left(2^{2l}+1,2^{4l+1}-1\right) = \gcd\left( 2^{2l}+1, 2^{2l+1}+1 \right) = \gcd\left( 2^{2l}+1,1 \right) = 1$$
and 
$$\gcd\left(2^{2l}+1, 2^{4l+1}+1\right) = \gcd\left( 2^{2l}+1, 2^{2l+1}-1 \right) = \gcd\left( 2^{2l}+1,3 \right) = 1,$$
where $l=\frac{k-1}{4}$. Thus there exist a unique positive integer $d\le 2^{2k}-1$  such that  $d\cdot (2^{\frac{k+1}{2}}+2) \equiv 1\pmod { 2^{2k}-1}$ and thus from $x^{2^{\frac{k+1}{2}}+2}=c^{2^{\frac{k+1}{2}}}$, we get a unique nonzero solution $x_0=\left(c^{2^{\frac{k+1}{2}}}\right)^d$. Moreover, it is easy to check $x_0=\left(c^{2^{\frac{k+1}{2}}}\right)^d$ is exactly the unique nonzero solution of Eq. \eqref{3.5} in this case. 

(II) When $a^{2q}\omega+a^2\neq0$ and $c+(u^q\omega + u)(a^q+a\omega)=0$, then it is trivial that $x=a^q+a\omega$ is a nonzero solution of Eq. \eqref{con1_eq1}. By direct computations, we can find that $x=a^q+a\omega$ is exactly a nonzero solution of Eq. \eqref{3.5} in this case. If $x\neq a^q+a\omega$, then by Eq. \eqref{con1_eq1}, we have 
$$y=\frac{cx+(u^q\omega+u)x^2}{\omega x^2+a^{2q}\omega+a^2}.$$
Plugging it into Eq. \eqref{con1_eq2} and simplifying, we get 
\begin{equation}
    \label{con1_eq3}
    Ax = c^{2^{\frac{k+1}{2}}+1} \omega,
\end{equation}
where 
$$A = c^{2^{\frac{k+1}{2}}} (u^q+u\omega^2) + c^2\omega + (u^2\omega + u^{2q}) (a^{2q}+a^2\omega^2). $$
Plugging $c = (u^q\omega + u)(a^q+a\omega)$ into the expression of $A$, we have $A = c^{2^{\frac{k+1}{2}}} (u^q+u\omega^2)$ and then Eq. \eqref{con1_eq3} becomes $(u^q+u\omega^2) x = c\omega$, implying $x = a^q+a\omega$. Thus, Eq. \eqref{3.5} has exactly one nonzero solution. 

(III) When  $a^{2q}\omega+a^2\neq0$ and $c+(u^q\omega + u)(a^q+a\omega)\neq0$, then $\omega x^2+a^{2q}\omega+a^2\neq0$, and it suffices to consider Eq. \eqref{con1_eq3}. In the following, we show that $A\neq0$ in this case and thus $x=c^{2^{\frac{k+1}{2}}+1} \omega A^{-1}$ is the unique nonzero solution of Eq. \eqref{con1_eq3}. If $A=0$, i.e.,
\begin{equation}
    \label{con1_eq4}
    c^{2^{\frac{k+1}{2}}} (u^q+u\omega^2) + c^2\omega + (u^2\omega + u^{2q}) (a^{2q}+a^2\omega^2) = 0,
\end{equation}
then raising Eq. \eqref{con1_eq4} to its $2^{\frac{k+1}{2}}$-th power and simplifying, we get 
$$  c^{2^{\frac{k+1}{2}}} = c(a+a^q\omega^2) + (a^{2q}+a^2\omega^2) (u^q+u\omega).$$
Plugging it into Eq. \eqref{con1_eq4}, we have 
\begin{equation}
    \label{con1_eq5}
    c^2+ \omega (a^q+\omega a) (u^q + \omega^2 u) c + \omega^2 (u^q+\omega^2 u) (u^qa^2+a^{2q}u) = 0.
\end{equation}
Let $\Delta=a^q+\omega a.$ Then $u^q+\omega^2u = \Delta^{2^{\frac{k+1}{2}}}$ and $$u^qa^2 +a^{2q}u = u(a^{2q}+\omega^2a^2) + a^2 (u^q+\omega^2u) = u\Delta + a^2 \Delta^{2^{\frac{k+1}{2}}}.$$ Moreover, Eq. \eqref{con1_eq5} becomes 
\begin{equation}
    \label{con1_eq6}
    c^2+\omega \Delta^{2^{\frac{k+1}{2}}+1}c+\omega^2\Delta^{2^{\frac{k+1}{2}}}\left( u\Delta^2 + a^2 \Delta^{2^{\frac{k+1}{2}}} \right) = 0.
\end{equation}
Since 
$$  \frac{ \omega^2\Delta^{2^{\frac{k+1}{2}}}\left( u\Delta^2 + a^2 \Delta^{2^{\frac{k+1}{2}}} \right)  }{ \omega^2 \Delta^{2^{\frac{k+1}{2}+1}+2} } = \frac{a^2}{\Delta^2} + \frac{a^{2^{\frac{k+1}{2}}}}{\Delta^{2^{\frac{k+1}{2}}}}, $$
by Eq. \eqref{con1_eq6}, we get 
\begin{equation}
\label{con1_eq7}
    c = \left( \frac{a^2}{\Delta^2} + \frac{a^4}{\Delta^4} + \cdots + \frac{a^{2^{\frac{k-1}{2}}}}{\Delta^{2^{\frac{k-1}{2}}}}    \right) \omega \Delta^{2^{\frac{k+1}{2}} +1}
\end{equation}
or 
\begin{equation}
    \label{con1_eq8}
    c = \left( \frac{a^2}{\Delta^2} + \frac{a^4}{\Delta^4} + \cdots + \frac{a^{2^{\frac{k-1}{2}}}}{\Delta^{2^{\frac{k-1}{2}}}}    +1 \right) \omega \Delta^{2^{\frac{k+1}{2}} +1}.
\end{equation}
  Raising Eq. \eqref{con1_eq7} into its $q$-th power and simplifying it by $\Delta^q = \omega^2\Delta$ and $\omega^{2^\frac{k+1}{2}}=\omega^2$, we obtain
$$ c^q = \left(  \frac{\omega^2a^{2q}}{\Delta^2} + \frac{\omega a^{4q}}{\Delta^4} + \cdots + \frac{\omega a^{2^{\frac{k-1}{2}}q}}{\Delta^{2^{\frac{k-1}{2}}}} \right)\omega^2\Delta^{2^{\frac{k+1}{2}}+1}.$$
Moreover, $c^q=c\omega$ implying 
$$\frac{\omega^2a^{2q}}{\Delta^2} + \frac{\omega a^{4q}}{\Delta^4} + \cdots + \frac{\omega a^{2^{\frac{k-1}{2}}q}}{\Delta^{2^{\frac{k-1}{2}}}} + \frac{a^2}{\Delta^2} + \frac{a^4}{\Delta^4} + \cdots + \frac{a^{2^{\frac{k-1}{2}}}}{\Delta^{2^{\frac{k-1}{2}}}} =0.$$

Let $\lambda = \frac{a+\omega a^{q}}{\Delta}$. Then $$\lambda^q=\frac{a^{q}+\omega^2 a}{\omega^2 \Delta} = \lambda.$$
Moreover, we have $$\lambda^2 + \lambda^4 + \cdots + \lambda^{2^{\frac{k-1}{2}}} = 0,$$
which implies 
$$ \left(\lambda^2 + \lambda^4 + \cdots + \lambda^{2^{\frac{k-1}{2}}}\right) + \left(\lambda^2 + \lambda^4 + \cdots + \lambda^{2^{\frac{k-1}{2}}}\right)^{\frac{1}{2}} = \lambda + \lambda^{2^{\frac{k-3}{2}}} = 0.$$
Since $\gcd\left( 2^{\frac{k-3}{2}}-1, 2^{k}-1   \right)=1$, we can get $\lambda\in\gf_2$. Then $\lambda^2+\lambda^4+\cdots + \lambda^{2^{\frac{k-1}{2}}}=0$ and thus $c=0$ or $c=\omega \Delta^{2^{\frac{k+1}{2}}+1}$ by Eq. \eqref{con1_eq7} and Eq. \eqref{con1_eq8}. Contradiction! (Note that in this case, 
$c \neq (u^q\omega + u)(a^q+a\omega) = \omega \Delta^{2^{\frac{k+1}{2}}+1}$.)

    In conclusion, for any $a \in \gf_{q^2}$, Eq. (\ref{3.5}) has exactly one nonzero solution in $\gf_{q^2}$, and thus $f(x)$ is $2$-to-$1$.
\end{proof}

In the following, we give this section's final infinite class of $2$-to-$1$ mappings. Before that, we provide a lemma. 

\begin{Lemma}
\label{lemma_conj2}
    Let $k$ be an odd positive integer and $q=2^k$, $t\in\gf_q\backslash\gf_2$. Then the equation $x^3+(1+t)x+t+t^2=0$ has three solutions in $\gf_{q^2}$, i.e.,
    $$x_1 = \alpha+\alpha^2, x_2 = \alpha\omega + \alpha^2\omega^2, x_3 = \alpha\omega^2 + \alpha^2\omega,$$
    where $\alpha\in\gf_q$ with $\alpha^3=1+t$ and $\omega\in\gf_4\backslash\gf_2$. 
\end{Lemma}

\begin{proof}
    By Lemma \ref{cubics soulution}, we first consider the quadratic equation $y^2+(t+t^2)y+(1+t)^3=0$, i.e.,
    \begin{equation}
    \label{lemma_eq1}
        \left(\frac{y}{t+t^2}\right)^2 + \frac{y}{t+t^2} + \frac{(1+t)^3}{(t+t^2)^2} = 0.
    \end{equation}
    Note that 
    $$\frac{(1+t)^3}{(t+t^2)^2} = \frac{1+t}{t^2} = \frac{1}{t} + \frac{1}{t^2}.$$
    Thus Eq. \eqref{lemma_eq1} has a solution $y_1 = \frac{t+t^2}{t} = 1+t.$ Moreover, since $k$ is odd, the equation $x^3=y_1$ has a exactly one solution in $\gf_q$, denoted by $\alpha$. That is to say, $\alpha^3 = 1+t$. Furthermore, it is trivial that $x^3=y_1$ has three solutions in $\gf_{q^2}$, i.e, $\alpha, \alpha\omega,\alpha\omega^2$. Finally, by Lemma \ref{cubics soulution}, we can find the three solutions in $\gf_{q^2}$ of the equation $x^3+(1+t)x+t+t^2=0$ as follows:
      $$x_1 = \alpha+\alpha^2, x_2 = \alpha\omega + \alpha^2\omega^2, x_3 = \alpha\omega^2 + \alpha^2\omega.$$
\end{proof}

\begin{Th}
    \label{conj2}
    Let $k$ be an odd positive integer and $q=2^k$, $c \in \gf_{q^2}^*$ satisfies $c^{q-1}=\omega$ with $\omega \in \gf_{2^2} \backslash \gf_2$. Then $f(x)=cx+\tr_{q^2/q}\left(x^{2^{2k-2}+2^{k-1}}\right)$ is $2$-to-$1$ over $\gf_{q^2}$.
\end{Th}
\begin{proof}
    Let $g(x)=f\left(x^4\right)=cx^4+\tr_{q^2/q}\left(x^{2^{2k}+2^{k+1}}\right)=cx^4+\tr_{q^2/q}\left(x^{2q+1}\right)$. According to Lemma \ref{$2$-to-$1$}, it suffices to show that for any $a\in \gf_{q^2}$, the equation $$g(x+a)+g(a)=0,$$ i.e.,  
    \begin{eqnarray}\label{ini3.6}
        cx^4+a^2x^q+a^qx^2+x^{q+2}+a^{2q}x+ax^{2q}+x^{2q+1}=0
    \end{eqnarray}
    has two solutions in $\gf_{q^2}$.
    Let $y=x^q$. Plugging it into Eq. (\ref{ini3.6}), we get
    \begin{eqnarray}\label{3.6-2}
        cx^4+a^qx^2+a^{2q}x+a^2y+x^2y+ay^2+xy^2=0.
    \end{eqnarray}
    Computing the sum of Eq. (\ref{3.6-2}) and its $q$-th power, we obtain $y=c^{\frac{1-q}{4}}x$, that is $y=\omega^{-\frac{1}{4}}x = \omega^2x$. Plugging the above relation into Eq. \eqref{3.6-2}, we have $x=0$ or
    \begin{equation}
        \label{eq1} x^3+\frac{x^2}{c}+\frac{a^q+a\omega}{c}x+\frac{a^2\omega^2+a^{2q}}{c}=0.
    \end{equation}
   
    Let $x=z+c^{-1}$. Then Eq. \eqref{eq1} becomes
    \begin{eqnarray}
   \label{labubu}
        z^3+ \left(\frac{1}{c^2}+\frac{a^q+a\omega}{c}\right)z+\frac{a^q+a\omega}{c^2}+\frac{a^2\omega^2+a^{2q}}{c}=0.
    \end{eqnarray}
 
  Let $s=a^qc$, $t=s+s^q$. Then $t = a^qc+ac^q=(a^q+a\omega)c$ and Eq. \eqref{labubu} becomes
    \begin{equation}
        \label{eq2}
        z_1^3 + (1+t)z_1+t+t^2=0,
    \end{equation}
    where $z_1=cz$.

    If $t=(a^q+a\omega)c=0$, i.e., $a=0$ or $a^{q-1}=\omega$, then $a=cu$ with $u\in\gf_q$. In this case, Eq. \eqref{eq2} becomes $z_1^3+z_1=0$, implying $z_1=0$ or $z_1=1$, i.e., $z=0$ or $z=c^{-1}$ and thus $x=c^{-1}$ or $x=0$. It is easy to check that in this case, $x=0,c^{-1}$ are indeed two solutions of Eq. \eqref{ini3.6}. 

    If $t=(a^q+a\omega)c=1$, then Eq. \eqref{eq2} becomes $z_1^3=0$, i.e., $z_1=0$ and then $x=c^{-1}$. It is also easy to check that in this case, $x=0,c^{-1}$ are indeed two solutions of Eq. \eqref{ini3.6}.  

    In the following, we assume that $t\in\gf_q\backslash\gf_2$. Then by Lemma \ref{lemma_conj2}, Eq. \eqref{eq2} has three solutions:
    $$z_1= \alpha+\alpha^2, \alpha\omega + \alpha^2\omega^2 ~\text{and}~\alpha\omega^2 + \alpha^2\omega,$$
    where $\alpha\in\gf_q$ with $\alpha^3=1+t$,
    and thus 
    \begin{equation}
    \label{eq3}
        x=c^{-1}z_1+c^{-1} = c^{-1}(\alpha+\alpha^2+1), c^{-1} (\alpha\omega + \alpha^2\omega^2 + 1) ~\text{and}~c^{-1}(\alpha\omega^2 + \alpha^2\omega + 1)
    \end{equation}
    are solutions of Eq. \eqref{eq1}. Note that the solutions of \eqref{3.6-2} satisfy the relation $x^q=\omega^2x$. It is easy to check that in the three solutions \eqref{eq3}, there is only one solution, i.e., $x=c^{-1}(\alpha+\alpha^2+1)$ that satisfies $x^q+\omega^2x$ and thus Eq. \eqref{ini3.6}  has only one nonzero solution in $\gf_{q^2}$. That is to say, Eq. \eqref{ini3.6} has two solutions, i.e., $x=0, c^{-1}(\alpha+\alpha^2+1)$ in $\gf_{q^2}$.
    
  In conclusion, for any $a\in\gf_{q^2}$, Eq. \eqref{ini3.6} has two solutions in $\gf_{q^2}$ and thus $f(x)$ is $2$-to-$1$. 
\end{proof}

\section{$2$-to-$1$ mappings of the form $cx^d+\tr_{q^l/q}(x^{d_1}+x^{d_2})$}\label{binomial}
This section focuses on constructing several $2$-to-$1$ mappings over the finite field $\gf_{q^l}$ of the form $cx^d+\tr_{q^l/q}(x^{d_1}+x^{d_2})$. The scope we search is as follows: $q=2^k$, $d_1,d_2 \in[1,q^l-2]$, $kl<14$, $c \in \gf_{q^l}^*$, $k,l \in \mathbb{N}^+$.
To narrow the search space, the following two additional conditions are imposed:
\begin{enumerate}[(1)]
	\item since \cite{charpin2009does} has already investigated the case where $k=1$, to avoid obtaining redundant results, it is assumed that $k>1$. 
	\item $d_1,d_2$ are the representative element of all positive integers less than $q^l-1$ under the cyclotomic equivalence. To prevent duplicate searches, it is assumed that $d_1<d_2$ and $\tr_{q^l/q}(x^{d_1}+x^{d_2}) \ne 0$.
        \item Regarding the range of $d$, only cases where $d$ within the range $[1,q^l-1]$ is a divisor of $q^l-1$ are considered according to Lemma \ref{divisor}.
\end{enumerate} 

Then we obtain all $2$-to-$1$ polynomials over $\gf_{q^l}$ of the form $cx^d+\tr_{q^l/q}(x^{d_1}+x^{d_2})$ with $q=2^k$, $d,d_1,d_2 \in[1,q^l-2]$, $kl<14$, $c \in \gf_{q^l}^*$ and $k,l \in \mathbb{N}^+$ in Table \ref{4-1}. The ``Reference" column in the table indicates
that this dataset served as an illustrative example for the theorem cited in the ``Reference”. The data illustrated that there does not exist a $2$-to-$1$ mapping of the form $cx^d+\tr_{q^l/q}(x^{d_1}+x^{d_2})$ when $d \ne 1$. Therefore, we consider the form $cx+\tr_{q^l/q}(x^{d_1}+x^{d_2})$. Most examples in Table \ref{4-1} can be generalized into the ``Reference".
\begin{table}[H]
	\centering
	\caption{$2$-to-$1$ polynomials of the form $cx+\tr_{q^l/q}(x^{d_1}+x^{d_2})$ in $\gf_{q^l}$, where $q=2^k$, $kl<14$}
	\label{4-1}
	\begin{threeparttable}
		\begin{tabular}{cccc}
			\toprule
			$(k,l)$& $(d_1,d_2)$& $c$& Reference\\
			\midrule
			\multirow{3}{*}{(2,3)}& \multirow{3}{*}{(5,10)}& \multirow{3}{*}{$\gf_4 \backslash \{0,1\}$}& Theorem \ref{q+1,2q+2}\\
			~& ~& ~& Theorem \ref{4.2}(\ref{q+1,q^2+q/2})\\
			~& ~& ~& Theorem \ref{4.2}(\ref{q^2+q/2+q^2+q/4})\\
			\midrule
			\multirow{2}{*}{(3,2)}& (13,15)& $\gf_4 \backslash \{0,1\}$& unknown\\
			~& (6,20)& $\gf_{2^6} \backslash \gf_{2^3}$& Theorem \ref{6,2q+4}\\
			\midrule
			\multirow{6}{*}{(3,3)}& (6,36)& $\gf_{2^3} \backslash \{0,1\}$& Theorem \ref{4.4}(\ref{6,q^2+q/2})\\
			~& (9,36)& $\gf_{2^3} \backslash \{0,1\}$& Theorem \ref{4.2}(\ref{q+1,q^2+q/2})\\
			~& (9,18)& $\delta^3,\delta^5, \delta^6$& unknown\\
			~& (18,36)& $\gf_{2^3} \backslash \{0,1\}$& Theorem \ref{4.2}(\ref{q^2+q/2+q^2+q/4})\\
			~& (34,36)& $\gf_{2^3} \backslash \{0,1\}$& Theorem \ref{4.4}(\ref{4q+2,q^2+q/2})\\
			~& (20,36)& $\gf_{2^3} \backslash \{0,1\}$& Theorem \ref{4.4}(\ref{2q+4,q^2+q/2})\\
			\bottomrule
		\end{tabular}
		\begin{tablenotes}
			\footnotesize
			\item[1] $\delta$ is primitive root in $\gf_{2^3}$.
		\end{tablenotes}
	\end{threeparttable}		
\end{table}

\begin{Th}
	\label{q+1,2q+2}
	Let $k$ be an even positive integer and $l$ be odd. Then $f(x)=cx+\tr_{q^l/q}(x^{q+1}+x^{2q+2})$ is $2$-to-$1$ over $\gf_{q^l}$, where $q=2^k$, $c\in \gf_{q}^*$ with $\tr_{q}(c^{-1}+1)=1$.
\end{Th}
\begin{proof}
	According to Lemma \ref{$2$-to-$1$}, it suffices to show that for any $a\in \gf_{q^l}$, the equation $$f(x+a)+f(a)=0,$$ i.e.,
	\begin{eqnarray}\label{8}
		cx+\tr_{q^l/q}(x^{q+1}+x^qa+a^qx+x^{2q+2}+x^{2q}a^2+a^{2q}x^2)=0
	\end{eqnarray}
	has two solutions in $\gf_{q^l}$.
    
    According to Eq. (\ref{8}), we get $cx=\tr_{q^l/q}(x^{q+1}+x^qa+a^qx+x^{2q+2}+x^{2q}a^2+a^{2q}x^2)\in \gf_{q}$ and $x\in \gf_{q}$ since $c\in \gf_{q}^*$. 
	
    From Lemma \ref{trace}(4), $$x^{q+1}+\tr_{q^l/q}(a^q)x+\tr_{q^l/q}(a)x^q+x^{2q+2}+\tr_{q^l/q}(a^2)x^{2q}+\tr_{q^l/q}(a^{2q})x^2+cx=0,$$
	where $l$ is odd and $x\in \gf_{q}$.
	\\In case of $\tr_{q^l/q}(a+a^q)=0$, we get 
	\begin{eqnarray}\label{9}
		x^4+x^2+cx=0.
	\end{eqnarray}
	
    Thus $x=0$ and it suffices to prove that for any $a \in \gf_{q^l}$, the equation $x^3+x+c=0$ has one solution in $\gf_{q^l}^*$.
	
    According to Lemma \ref{cubic}, $f_1(x)=x^3+x+c$, $\tr_{q}(\frac{1}{c^2})\ne \tr_{q}(1)$ since $c\ne 0$ and $\tr_{q} (c^{-1}+1)=1$, then $f_1=(1,2)$. Therefore, Eq. (\ref{9}) has two solutions in $\gf_{q}$. Moreover, Eq. (\ref{9}) has two solutions in $\gf_{q^l}$, $f(x)$ is $2$-to-$1$ over $\gf_{q^l}$.
\end{proof}
\begin{Th}
\label{4.2}
	Let $l$ be an odd positive integer and $q=2^k$, $k \in \mathbb{N}^+$, $c \in \gf_{q}\backslash \{0,1\}$. Then
	\begin{enumerate}[1)]
		\item $f(x)=cx+\tr_{q^l/q}\left(x^{q+1}+x^{\frac{q^2+q}{2}}\right)$;\label{q+1,q^2+q/2}
		\item $f(x)=cx+\tr_{q^l/q}\left(x^{\frac{q^2+q}{2}}+x^{\frac{q^2+q}{4}}\right)$ \label{q^2+q/2+q^2+q/4}
	\end{enumerate}
	are $2$-to-$1$ over $\gf_{q^l}$. 
\end{Th}
\begin{proof}
In the first (resp. second) function, we give the substitution $x \mapsto x^2$ (resp. $x^4$). Since the proof techniques for these cases are analogous, we present a detailed demonstration for the first case as a representative example.
	
    Let $g(x)=f(x^2)=cx^2+\tr_{q^l/q}(x^{2q+2}+x^{q^2+q})$. It suffices to prove that for any $d\in \gf_{q^l}$, the equation $g(x)=d$ has 0 or 2 solutions in $\gf_{q^l}$. Then we consider the solution of the following equation, 
	\begin{eqnarray}\label{10}
		cx^2+\tr_{q^l/q}(x^{2q+2}+x^{q^2+q})=d.
	\end{eqnarray}
	
    Let $u=cx^2+d=\tr_{q^l/q}(x^{q^2+q})+\tr_{q^l/q}(x^{2q+2})\in \gf_{q}$. Then $x=\left(\frac{u+d}{c}\right)^{\frac{1}{2}}$ due to $c\ne0$. Moreover, since $u \in \gf_{q}$ and $c \in \gf_{q}\verb|\|\{0,1\}$, we have $$x^{q^2+q}=\left(\frac{u+d}{c}\right)^{\frac{q^2+q}{2}}=\frac{(u^2+(d^{q^2}+d^q)u+d^{q^2+q})^\frac{1}{2}}{c},$$ $$x^{2q+2}=\left(\frac{u+d}{c}\right)^{q+1}=\frac{u^2+(d^{q^2}+d)u+d^{q+1}}{c^2}.$$
	
    Plugging the above two equations into Eq. (\ref{10}), we get
	\begin{eqnarray*}
		u&=&\tr_{q^l/q}\left(\frac{u^2+u\left(d+d^q\right)+d^{q+1}}{c^2}+\frac{\left(u^2+\left(d^{q^2}+d^q\right)u+d^{q^2+q}\right)^\frac{1}{2}}{c}\right)\\
		&=&\frac{u^2}{c^2}+\tr_{q^l/q}\left(d+d^q\right)\frac{u}{c^2}+\frac{\tr_{q^l/q}\left(d^{q+1}\right)}{c^2}+\frac{u}{c}+\left(\tr_{q^l/q}\left(d^{q^2}+d^q\right)\right)^\frac{1}{2}\frac{u^\frac{1}{2}}{c}+\frac{\left(\tr_{q^l/q}\left(d^{q^2+q}\right)\right)^\frac{1}{2}}{c},
	\end{eqnarray*}
	where the second equality holds due to Lemma \ref{trace}(4) and $l$ being odd. Note that for any $d \in \gf_{q^l}$, it is clear that $\tr_{q^l/q}(d+d^q)=0$ and thus the above equation becomes 
	$$u=\frac{u^2}{c^2}+\frac{\tr_{q^l/q}\left(d^{q+1}\right)}{c^2}+\frac{u}{c}+\frac{\left(\tr_{q^l/q}\left(d^{q^2+q}\right)\right)^\frac{1}{2}}{c},$$
	that is 
	\begin{eqnarray}\label{11}
		u^2+(c^2+c)u+c\left(\tr_{q^l/q}\left(d^{q^2+q}\right)\right)^{\frac{1}{2}}+\tr_{q^l/q}(d^{q+1})=0
	\end{eqnarray}
	
    Since $c \notin \{0,1\}$, Eq. (\ref{11}) has 0 or 2 solutions in $\gf_{q}$. Hence $g(x)=d$ has 0 or 2 solutions in $\gf_{q^l}$, that is $g(x)$ is $2$-to-$1$ over $\gf_{q^l}$, $f(x)$ is $2$-to-$1$ over $\gf_{q^l}$.
\end{proof}
\begin{Th}
	\label{6,2q+4}
	Let $k$ be an odd positive integer and $q=2^k$, $c \in \gf_{q^2}$ satisfies $c^{q-1}=\omega$, where $\omega^3=1$. Then $f(x)=cx+\tr_{q^2/q}(x^6+x^{2q+4})$ is $2$-to-$1$ over $\gf_{q^2}$.
\end{Th}
\begin{proof}
	According to Lemma \ref{$2$-to-$1$}, it suffices to show that for any $a \in \gf_{q^2}$, the equation $$f(x+a)+f(a)=0,$$ i.e., 
	\begin{eqnarray}\label{12}
		cx+\tr_{q^2/q}((x+a)^6+a^6+(x+a)^{2q+4}+a^{2q+4})=0
	\end{eqnarray}
	has two solutions in $\gf_{q^2}$.
    
     Let $y=x^{2q}$, $b=a^{2q}$. From Eq. (\ref{12}), we get 
	\begin{eqnarray}\label{13}
		\left(a^2+a^{2q}\right)x^4+\left(a^4+a^{4q}\right)x^2+cx+x^4y+x^4b+a^4y+x^2y^2+x^2b^2+a^2y^2=0.
	\end{eqnarray}
	 
     According to Eq. (\ref{12}), we get $cx=\tr_{q^2/q}((x+a)^6+a^6+(x+a)^{2q+4}+a^{2q+4})\in \gf_{q}$. Then $y=c^{2(1-q)}x^2$, implying $y=\omega x^2$.
    Then from Eq. (\ref{13}),$$\left(a^2+a^{2q}\right)x^4+\left(a^4+a^{4q}\right)x^2+\left(\omega+\omega^2\right)x^6+\left(a^2\omega^2+b\right)x^4+\left(a^4\omega+b^2\right)x^2+cx=0.$$
	
    According to the properties of cubic primitive roots $\omega^2+\omega+1=0$, the above equation becomes $$x^6+\left(a^2\omega^2+b\right)x^4+\left(a^4\omega+b^2\right)x^2+cx=0.$$
	
    Obviously $x=0$, it only need to prove for any $a\in \gf_{q^2}$, the equation 
	\begin{eqnarray}\label{14}
		x^5+\left(a^2\omega^2+b\right)x^3+\left(a^4\omega+b^2\right)x+c=0
	\end{eqnarray}
	has one solution in $\gf_{q^2}^*$.
	
    We notice that $x^5+\left(a^2\omega^2+b\right)x^3+\left(a^4\omega+b^2\right)x=D_5\left(x,a^2\omega^2+b\right)$ is Dickson polynomial of degree 5. Recalling that $k$ is odd, then $\mathrm{gcd}(5,q^2-1)=1$. From Lemma \ref{gcd}, $D_5(x,a^2\omega^2+b)$ permutes $\gf_{q}$. That is, Eq. (\ref{14}) only has one solution in $\gf_{q}$, and the solution is not 0 since $c \ne 0$. Therefore, the Eq. (\ref{12}) has two solutions in $\gf_{q^2}$ and $f(x)$ is $2$-to-$1$ over $\gf_{q^2}$.
\end{proof}
\begin{Th}
\label{4.4}
	Let $k$ and $l$ be odd positive integers and $q=2^k$, $c \in \gf_{q}\verb|\|\{0,1\}$. Then
	\begin{enumerate}[1)]
		\item $f(x)=cx+\tr_{q^l/q}\left(x^6+x^{\frac{q^2+q}{2}}\right)$;\label{6,q^2+q/2}
		\item $f(x)=cx+\tr_{q^l/q}\left(x^{4q+2}+x^{\frac{q^2+q}{2}}\right)$;\label{4q+2,q^2+q/2}
		\item $f(x)=cx+\tr_{q^l/q}\left(x^{2q+4}+x^{\frac{q^2+q}{2}}\right)$ \label{2q+4,q^2+q/2}
	\end{enumerate}
	are $2$-to-$1$ over $\gf_{q^l}$. 
\end{Th}
\begin{proof}
Since the proof techniques for these cases are analogous, we present a detailed demonstration for the first case as a representative example.
	
    According to Lemma \ref{$2$-to-$1$}, it suffices to show that for any $a \in \gf_{q^l}$, the equation $$f(x+a)+f(a)=0,$$ i.e.,
	\begin{eqnarray}\label{15}
		cx+\tr_{q^l/q}(x^6+a^2x^4+a^4x^2+(x^{q^2+q}+x^{q^2}a^q+a^{q^2}x^q)^\frac{1}{2})=0
	\end{eqnarray}
	has two solutions in $\gf_{q^l}$. 
    
    According to Eq. (\ref{15}), we get $$cx=\tr_{q^l/q}(x^6+a^2x^4+a^4x^2+(x^{q^2+q}+x^{q^2}a^q+a^{q^2}x^q)^\frac{1}{2})\in \gf_{q}$$ and then $x \in \gf_{q}$ since $c \in \gf_{q^l}^*$. 
    
    From Lemma \ref{trace}(4), $$x^6+\tr_{q^l/q}(a^2)x^4+\tr_{q^l/q}(a^4)x^2+(x^2+\tr_{q^l/q}(a^q)x+\tr_{q^l/q}(a^{q^2})x)^\frac{1}{2}+cx=0,$$
	where $l$ is odd and $x\in \gf_{q}$.
	
    In case of $\tr_{q^l/q}(a+a^q)=0$, we get 
	\begin{eqnarray}\label{16}
		x^6+\tr_{q^l/q}(a^2)x^4+\tr_{q^l/q}(a^4)x^2+x+cx=0.
	\end{eqnarray}
	
    Thus $ x=0$ or 
	\begin{eqnarray}\label{17}
		x^5+\tr_{q^l/q}(a^2)x^3+\tr_{q^l/q}(a^4)x+1+c=0.
	\end{eqnarray}
	
    We notice that $$x^5+\tr_{q^l/q}(a^2)x^3+\tr_{q^l/q}(a^4)x=D_5(x,\tr_{q^l/q}(a^2))$$ is a Dickson polynomial of degree 5. Recalling that $k$ is odd, then $\mathrm{gcd}(5,q^2-1)=1$. From Lemma \ref{gcd}, $D_5(x,\tr_{q^l/q}(a^2))$ permutes $\gf_{q}$. That is, Eq. (\ref{17}) has only one solution in $\gf_{q}$ and the solution is non-zero since $c \notin \{0,1\}$. Therefore, the Eq. (\ref{15}) has two solutions in $\gf_{q^l}$ and $f(x)$ is $2$-to-$1$ over $\gf_{q^l}$.
\end{proof}
\section{binary linear codes from some new $2$-to-$1$ functions}\label{linear}
In this section, we apply the $2$-to-$1$ mappings obtained in Section \ref{monomial} to construct binary linear codes $C_f$ defined as in \eqref{cf}.

\begin{table}[H]
		\centering
		\caption{The weight distribution of linear codes in Theorem \ref{ex}}
		\label{code1}
        \begin{threeparttable}
			\begin{tabular}{cc}
				\toprule
				weight& multiplicity \\
				\midrule
				0& 1\\
                $2^{n-1}$& $2^{n-k}+2^{n+k}-2^n-1$\\
                $2^{n-1}-2^{\frac{n+k}{2}-1}$& $2^{\frac{n+k}{2}-1}-2^{\frac{n-k}{2}-1}+2^{n-1}-2^{n-k-1}$\\
                $2^{n-1}+2^{\frac{n+k}{2}-1}$& $2^{\frac{n-k}{2}-1}-2^{\frac{n+k}{2}-1}+2^{n-1}-2^{n-k-1}$\\
				\bottomrule
			\end{tabular}
			
		\end{threeparttable}		
	\end{table}

\begin{Th}\label{ex}
    Let $l$ be an odd positive integer, $q=2^k$, $k \in \mathbb{N}$, $n=kl$,  $c\in\mathbb{F}_q^\ast$, 
	\begin{enumerate}[1)]
		\item $f(x)=cx+\tr_{q^l/q}(x^{q+1})$;
		\item $f(x)=cx+\tr_{q^l/q}(x^{2q+2})$, where $k$ is odd; 
		\item $f(x)=cx+\tr_{q^l/q}(x^{q^2+1})$. 
	\end{enumerate} Define linear code $C_f$ as in (\ref{cf}) is a $[2^n-1, n+k,2^{n-1}-2^{\frac{n+k}{2}-1}]$ binary linear code with weight distribution in Table \ref{code1}.
\end{Th}
\begin{proof}
    The linear codes constructed from these three functions have the same weight distribution. Now we take the first one $f(x)=cx+\tr_{q^l/q}(x^{q+1})$ as an example to construct the linear code $C_f$.
    We compute the Walsh transforms $W_f(a.b)$ defined as in (\ref{walsh}). 
    \begin{eqnarray*}
    \tr_{2^n}(ax+bf(x))
    &=&\tr_{2^n}((a+bc)x+b\tr_{2^n/2^k}(x^{q+1}))\\
    &=&\tr_{2^n}(ax+bcx+b(x^{q+1}+x^{(q+1)q}+\cdots+x^{(q+1)q^{l-1}}))\\
    &=&\tr_{2^n}(ax+bcx+bx^{q+1}+(b^{\frac{1}{q}}x^{q+1})^q+\cdots+(b^{\frac{1}{q^{l-1}}}x^{q+1})^{q^{l-1}})\\
    &=&\tr_{2^n}(ax+bcx+(b+b^\frac{1}{q}+\cdots+b^\frac{1}{q^{l-1}})x^{q+1})\\
    &=&\tr_{2^n}((a+bc)x+\tr_{2^n/2^k}(b^\frac{1}{q^{l-1}})x^{q+1})\\
    &=&\tr_{2^n}((a+bc)x+\tr_{2^n/2^k}(b)x^{q+1}).
    \end{eqnarray*}
    Then we have the Walsh transform 
     $$W_f(a,b)=\sum_{x\in \gf_{2^n}}(-1)^{\tr_{2^n}(ax+bf(x))}=\sum_{x\in \gf_{2^n}}(-1)^{\tr_{2^n}((a+bc)x+\tr_{2^n/2^k}(b)x^{q+1})}.$$
     
When $\tr_{2^n/2^k}(b)=0$,
$$W_f(a,b)=\sum_{x\in \gf_{2^n}}(-1)^{\tr_{2^n}(ax)}=\begin{cases}
2^n &a+bc=0 ,\\
0 &a+bc\ne0.
\end{cases}
$$

When $\tr_{2^n/2^k}(b)\ne 0$, $\varphi_{a,b}(x)=\tr_{2^n}((a+bc)x+\tr_{2^n/2^k}(b)x^{q+1})$, the bilinear form of $\varphi_{a,b}(x)$ is given by
\begin{eqnarray*}
    B_{\varphi_{a,b}}(x,y)
    &=&\varphi_{a,b}(x+y)+\varphi_{a,b}(x)+\varphi_{a,b}(y)\\
    &=&\tr_{2^n}(\tr_{2^n/2^k}(b)(x+y)^{q+1}+\tr_{2^n/2^k}(b)x^{q+1}+\tr_{2^n/2^k}(b)y^{q+1})\\
    &=&\tr_{2^n}(\tr_{2^n/2^k}(b)x^qy+\tr_{2^n/2^k}(b)xy^q)\\
    &=&\tr_{2^n}((\tr_{2^n/2^k}(b)y+(\tr_{2^n/2^k}(b))^qy^{q^2})x^q).
\end{eqnarray*}

Its kernel is given by 
\begin{eqnarray*}
    V_{\varphi_{a,b}}
    &=&\{y \in \gf_{2^n}: B_{\varphi_{a,b}}(x,y)=0 ,\forall x \in \gf_{2^n}\}\\
    &=&\{y \in \gf_{2^n}: \tr_{2^n/2^k}(b)y+(\tr_{2^n/2^k}(b))^qy^{q^2}=0\},
\end{eqnarray*}
the number of solutions are $\gcd(q^2-1,q^l-1)+1=q=2^k$, then the dimension of $V_{\varphi_{a,b}}$ is $d=\log_2{(\gcd(q^2-1,q^l-1)+1)}=k$.

From Lemma \ref{qua}, we have
$$W_f(a,b)
\begin{cases}
\pm2^{\frac{n+k}{2}}, & \text{if $\varphi_{a,b}$ vanishes on $V_{\varphi_{a,b}}$}, \\
0, & \text{otherwise}.  
\end{cases}$$

Generally, $W_f(a,b)=\{2^n,0,\pm 2^{\frac{n+k}{2}} \}$. Let their corresponding frequencies be $X_0, X_1, X_2, X_3$.
When $a+bc=0$ and $\tr_{2^n/2^k}(b)=0$, $W_f(a,b)=2^n$. The number of solutions of $\tr_{2^n/2^k}(b)=0$ in $\gf_{2^n}$ is $2^{n-k}$ and $b=\frac{a}{c}\in \gf_{2^n}$, then the $X_0=2^{n-k}$. We have  
$$\begin{cases}
    2^{n-k}+X_1+X_2+X_3=2^{2n}\\
    2^{2n-k}+2^{\frac{n+k}{2}}X_2-2^{\frac{n+k}{2}}X_3=2^{2n}\\
    2^{3n-k}+2^{n+k}X_2+2^{n+k}X_3=2^{3n}.
\end{cases}
$$
Then $$\begin{cases}
    X_0=2^{n-k}\\
    X_1=2^{2n}+2^{2n-2k}-2^{2n-k}-2^{n-k}\\
    X_2=2^{\frac{3n-k}{2}-1}-2^{\frac{3n-3k}{2}-1}+2^{2n-k-1}-2^{2n-2k-1}\\
    X_3=2^{2n-k-1}-2^{2n-2k-1}-2^{\frac{3n-k}{2}-1}+2^{\frac{3n-3k}{2}-1}.
    
\end{cases}
$$
The weight distribution of the linear code $C_f$ is entirely determined by the Walsh spectrum of $f$. If a value 
$W_f(a,b)$ appears $X$ times in the Walsh spectrum of $f$, then the number of codewords in $C_f$ with Hamming weight $2^{n-1}-\frac{1}{2}W_f(a,b)$ is $\frac{X}{X_0}$. Therefore, in the linear code $C_f$, the Hamming weight of a codeword $c_{a,b}$ satisfies $wt(c_{a,b})=\{0,2^{n-1},2^{n-1}-2^{\frac{n+k}{2}-1},2^{n-1}+2^{\frac{n+k}{2}-1}\}$. The weight distribution is shown in Table \ref{code1}...
\end{proof}
\begin{table}[H]
		\centering
		\caption{The weight distribution pf linear codes in Theorem \ref{ex1}}
		\label{code2}
        \begin{threeparttable}
			\begin{tabular}{cc}
				\toprule
				weight& multiplicity \\
				\midrule
				0& 1\\
                $2^{n-1}$& $2^{n-1}+2^{n+k}-2^{n-1+k}-1$\\
                $2^{n-1}-2^{\frac{n+1}{2}-1}$& $2^{\frac{n+1}{2}-2}-2^{\frac{n+1}{2}-2+k}+2^{n-2+k}-2^{n-2}$\\
                $2^{n-1}+2^{\frac{n+1}{2}-1}$& $2^{\frac{n+1}{2}-2+k}-2^{\frac{n+1}{2}-2}+2^{n-2+k}-2^{n-2}$\\
				\bottomrule
			\end{tabular}
			
		\end{threeparttable}		
	\end{table}

\begin{Th}
\label{ex1}
    Let $n=kl$, $k,l$ be odd, $c \in \gf_{2^n}^*$, $q=2^k$. 
    \begin{enumerate}[1)]
		\item $f(x)=cx+\tr_{q^l/q}(x^6)$;
		\item $f(x)=cx+\tr_{q^l/q}(x^{4q+2})$; 
		\item $f(x)=cx+\tr_{q^l/q}(x^{2q+4})$. 
	\end{enumerate} $C_f$ is a $[2^n-1, n+k,2^{n-1}-2^{\frac{n+1}{2}-1}]$ binary linear code with weight distribution in Table \ref{code2}.
\end{Th}
\begin{proof}
    The poof process is similar to Theorem \ref{ex}.
\end{proof}

\begin{table}[H]
		\centering
		\caption{The weight distribution of linear codes in Theorem \ref{ex2}}
		\label{code3}
        \begin{threeparttable}
			\begin{tabular}{cc}
				\toprule
				weight& multiplicity \\
				\midrule
				0& 1\\
                $2^{n-1}$& $2^{n-2}+2^{n+k}-2^{n-2+k}-1$\\
                $2^{n-1}-2^{\frac{n+2}{2}-1}$& $2^{\frac{n}{2}-2+k}-2^{\frac{n}{2}-2}+2^{n-3+k}-2^{n-3}$\\
                $2^{n-1}+2^{\frac{n+2}{2}-1}$& $2^{\frac{n}{2}-2}-2^{\frac{n}{2}-2+k}+2^{n-3+k}-2^{n-3}$\\
				\bottomrule
			\end{tabular}
			
		\end{threeparttable}		
	\end{table}

\begin{Th}
\label{ex2}
    Let $f(x)=cx+\tr_{q^2/q}(x^{2q+4})$, $n=2k$, $k$ be odd, $q=2^k$, $c^{q-1}=\omega$ where $\omega^3=1$. $C_f$ is a $[2^n-1, n+k, 2^{n-1}-2^{\frac{n+2}{2}-1}]$ binary linear code with weight distribution in Table \ref{code3}.
\end{Th}
\begin{proof}
    The poof process is similar to Theorem \ref{ex}.
\end{proof}

\begin{Rem}
    From the weight distributions (Tables \ref{code1}, \ref{code2}, and \ref{code3}) of linear codes in this section, we can see that these linear codes are all self-orthogonal and minimal since all weights are divisible by $4$ and $\frac{w_{\min}}{w_{\max}}>\frac{1}{2}$. 
\end{Rem}

\section{Conclusion}

\label{conclusion}
In this paper, we primarily studied $2$-to-$1$ polynomials of the form $F(x)=G(x)+\tr_{q^l/q}(R(x))$, where  $q=2^k$, $G(x)$ is a monomial and $R(x)$ is a monomial or binomial. Starting from the properties of $2$-to-$1$ mappings, we first utilized MAGMA to search for specific $2$-to-$1$ polynomials data generated by different values of $k$ and $l$ over small fields. By employing the elementary method, combining the properties of Dickson polynomials, and the Dobbertin's multivariate method, we constructed $16$ new infinite classes of $2$-to-$1$ polynomials, which include $9$ classes of $R(x)$ that are monomials and $7$ classes of $R(x)$ that are binomials. Finally, some binary self-orthogonal and minimal linear codes with few weights are provided, implying our constructions of $2$-to-$1$ mappings are useful. Note that the $2$-to-$1$ mappings proposed in this paper can also be applied to constructing balanced Boolean functions as done in  \cite{qu2025parametric}.


\begin{thebibliography}{10}

    \bibitem{akbary2011constructing}
    Amir Akbary, Dragos Ghioca, and Qiang Wang.
    \newblock On constructing permutations of finite fields.
    \newblock {\em Finite Fields and Their Applications}, 17(1):51--67, 2011.
    
    \bibitem{ashikhmin2002minimal}
    Alexei Ashikhmin and Alexander Barg.
    \newblock Minimal vectors in linear codes.
    \newblock {\em IEEE Transactions on Information Theory}, 44(5):2010--2017, 1998.
    
    \bibitem{bartoli2022two}
    Daniele Bartoli, Massimo Giulietti, and Marco Timpanella.
    \newblock Two-to-one functions from galois extensions.
    \newblock {\em Discrete Applied Mathematics}, 309:194--201, 2022.
    
    \bibitem{berger2003application}
    Toby Berger and Vladimir~I Levenshtein.
    \newblock Application of cover-free codes and combinatorial designs to two-stage testing.
    \newblock {\em Discrete Applied Mathematics}, 128(1):11--26, 2003.
    
    \bibitem{berlekamp1967solution}
    Elwyn~R. Berlekamp, Howard Rumsey, and Gustave Solomon.
    \newblock On the solution of algebraic equations over finite fields.
    \newblock {\em Information and Control}, 10(6):553--564, 1967.
    
    \bibitem{browning2010apn}
    Keith~A Browning, John~F Dillon, MT~McQuistan, and Alan~J Wolfe.
    \newblock An {APN} permutation in dimension six.
    \newblock {\em Finite Fields: theory and applications}, 518:33--42, 2010.
    
    \bibitem{calderbank1997quantum}
    A~Robert Calderbank, Eric~M Rains, Peter~W Shor, and Neil~JA Sloane.
    \newblock Quantum error correction and orthogonal geometry.
    \newblock {\em Physical Review Letters}, 78(3):405, 1997.
    
    \bibitem{carlet2005linear}
    Claude Carlet, Cunsheng Ding, and Jin Yuan.
    \newblock Linear codes from perfect nonlinear mappings and their secret sharing schemes.
    \newblock {\em IEEE Transactions on Information Theory}, 51(6):2089--2102, 2005.
    
    \bibitem{carlet2011dillonʼs}
    Claude Carlet and Sihem Mesnager.
    \newblock On {D}illon's class ${H}$ of bent functions, {N}iho bent functions and o-polynomials.
    \newblock {\em Journal of Combinatorial Theory, Series A}, 118(8):2392--2410, 2011.
    
    \bibitem{cesmelioglu2015bent}
    Ayca Cesmelioglu, Wilfried Meidl, and Alexander Pott.
    \newblock Bent functions, spreads, and o-polynomials.
    \newblock {\em SIAM Journal on Discrete Mathematics}, 29(2):854--867, 2015.
    
    \bibitem{charpin2009does}
    Pascale Charpin and Gohar Kyureghyan.
    \newblock When does {$G (x)+ \gamma \tr(H(x))$} permute $\gf_{p^n}$?
    \newblock {\em Finite Fields and Their Applications}, 15(5):615--632, 2009.
    
    \bibitem{charpin2010monomial}
    Pascale Charpin and Gohar Kyureghyan.
    \newblock Monomial functions with linear structure and permutation polynomials.
    \newblock {\em Finite fields: theory and applications}, 518:99--111, 2010.
    
    \bibitem{chen2017equivalent}
    Xi~Chen, Yazhi Deng, Min Zhu, and Longjiang Qu.
    \newblock An equivalent condition on the switching construction of differentially 4-uniform permutations from the inverse function.
    \newblock {\em International Journal of Computer Mathematics}, 94(6):1252--1267, 2017.
    
    \bibitem{cherowitzo1988hyperovals}
    William Cherowitzo.
    \newblock Hyperovals in {D}esarguesian planes of even order.
    \newblock In {\em Annals of Discrete Mathematics}, volume~37, pages 87--94. Elsevier, 1988.
    
    \bibitem{delsarte1998association}
    Philippe Delsarte and Vladimir~I. Levenshtein.
    \newblock Association schemes and coding theory.
    \newblock {\em IEEE Transactions on Information Theory}, 44(6):2477--2504, 1998.
    
    \bibitem{dillon2006apn}
    John~F Dillon.
    \newblock {APN} polynomials and related codes.
    \newblock In {\em Banff Conference}, 2006.
    
    \bibitem{ding2007cyclotomic}
    Cunsheng Ding and Harald Niederreiter.
    \newblock Cyclotomic linear codes of order $3$.
    \newblock {\em IEEE Transactions on Information Theory}, 53(6):2274--2277, 2007.
    
    \bibitem{dobbertin2002uniformly}
    Hans Dobbertin.
    \newblock Uniformly representable permutation polynomials.
    \newblock In {\em Sequences and their Applications: Proceedings of SETA’01}, pages 1--22. Springer, 2002.
    
    \bibitem{edel2008new}
    Yves Edel and Alexander Pott.
    \newblock A new almost perfect nonlinear function which is not quadratic.
    \newblock {\em Advances in Mathematics of Communications}, 3(1):59--81, 2009.
    
    \bibitem{helleseth2005error}
    Tor Helleseth, T~Klove, and Vladimir~I Levenshtein.
    \newblock Error-correction capability of binary linear codes.
    \newblock {\em IEEE Transactions on Information Theory}, 51(4):1408--1423, 2005.
    
    \bibitem{huffman2010fundamentals}
    W~Cary Huffman and Vera Pless.
    \newblock {\em Fundamentals of error-correcting codes}.
    \newblock Cambridge University Press, 2010.
    
    \bibitem{idrisova2019algorithm}
    Valeriya Idrisova.
    \newblock On an algorithm generating 2-to-1 {APN} functions and its applications to “the big {APN} problem”.
    \newblock {\em Cryptography and Communications}, 11(1):21--39, 2019.
    
    \bibitem{ShaJiang2025New}
    Sha Jiang, Mu~Yuan, Kangquan Li, and Longjiang Qu.
    \newblock New constructions of permutation polynomials of the form {$x+\gamma \tr_q^{q^2}(h(x))$} over finite fields with even characteristic.
    \newblock {\em Finite Fields and Their Applications}, 101:102522, 2025.
    
    \bibitem{kolsch2024classifications}
    Lukas K{\"o}lsch and Gohar Kyureghyan.
    \newblock The classifications of o-monomials and of 2-to-1 binomials are equivalent.
    \newblock {\em Designs, Codes and Cryptography}, pages 1--10, 2024.
    
    \bibitem{kyureghyan2011constructing}
    Gohar~M Kyureghyan.
    \newblock Constructing permutations of finite fields via linear translators.
    \newblock {\em Journal of Combinatorial Theory, Series A}, 118(3):1052--1061, 2011.
    
    \bibitem{leonard1972quartics}
    Philip~A Leonard and Kenneth~S Williams.
    \newblock Quartics over {${\rm GF}(2^n)$}.
    \newblock {\em Proceedings of the American Mathematical Society}, pages 347--350, 1972.
    
    \bibitem{li2022dillon}
    Chunlei Li, Constanza Riera, and Pantelimon Stanica.
    \newblock Low $c$-differentially uniform functions via an extension of {Dillon}'s switching method.
    \newblock {\em arXiv preprint arXiv:2204.08760}, 2022.
    
    \bibitem{li2021binary}
    Kangquan Li, Chunlei Li, Tor Helleseth, and Longjiang Qu.
    \newblock Binary linear codes with few weights from two-to-one functions.
    \newblock {\em IEEE Transactions on Information Theory}, 67(7):4263--4275, 2021.
    
    \bibitem{li2021furtherstudy}
    Kangquan Li, Sihem Mesnager, and Longjiang Qu.
    \newblock Further study of 2-to-1 mappings over $\gf_{2^n}$.
    \newblock {\em IEEE Transactions on Information Theory}, 67(6):3486--3496, 2021.
    
    \bibitem{li2018permutationandtrinomial}
    Kangquan Li, Longjiang Qu, Xi~Chen, and Li~Chao.
    \newblock Permutation polynomials of the form $cx+ \tr_{q^l/q}(x^a)$ and permutation trinomials over finite fields with even characteristic.
    \newblock {\em Cryptography and Communications}, 10:531--554, 2018.
    
    \bibitem{Li2017NewConstructions}
    Kangquan Li, Longjiang Qu, and Qiang Wang.
    \newblock New constructions of permutation polynomials of the form $x^rh\left(x^{q-1}\right)$ over $\mathbb{F}_{q^2}$.
    \newblock {\em Designs, Codes and Cryptography}, 86:2379--2405, 2017.
    
    \bibitem{li2019compositional}
    Kangquan Li, Longjiang Qu, and Qiang Wang.
    \newblock Compositional inverses of permutation polynomials of the form $x^rh(x^s)$ over finite fields.
    \newblock {\em Cryptography and Communications}, 11:279--298, 2019.
    
    \bibitem{lidl1993theory}
    Rudolf Lidl.
    \newblock Theory and applications of {Dickson} polynomials.
    \newblock In {\em Topics In Polynomials Of One And Several Variables And Their Applications: Volume Dedicated to the Memory of PL Chebyshev (1821--1894)}, pages 371--395. World Scientific, 1993.
    
    \bibitem{lidl1997finite}
    Rudolf Lidl and Harald Niederreiter.
    \newblock {\em Finite fields}.
    \newblock Cambridge University Press, 1997.
    
    \bibitem{massey1993minimal}
    James~L Massey.
    \newblock Minimal codewords and secret sharing.
    \newblock In {\em Proceedings of the 6th joint Swedish-Russian international workshop on information theory}, pages 276--279, 1993.
    
    \bibitem{mesnager2014several}
    Sihem Mesnager.
    \newblock Several new infinite families of bent functions and their duals.
    \newblock {\em IEEE Transactions on Information Theory}, 60(7):4397--4407, 2014.
    
    \bibitem{mesnager2023several}
    Sihem Mesnager, Liqin Qian, Xiwang Cao, and Mu~Yuan.
    \newblock Several families of binary minimal linear codes from two-to-one functions.
    \newblock {\em IEEE Transactions on Information Theory}, 69(5):3285--3301, 2023.
    
    \bibitem{mesnager2019two}
    Sihem Mesnager and Longjiang Qu.
    \newblock On two-to-one mappings over finite fields.
    \newblock {\em IEEE Transactions on Information Theory}, 65(12):7884--7895, 2019.
    
    \bibitem{mesnager2022more}
    Sihem Mesnager, Mu~Yuan, and Dabin Zheng.
    \newblock More about the corpus of involutions from two-to-one mappings and related cryptographic {S-boxes}.
    \newblock {\em IEEE Transactions on Information Theory}, 69(2):1315--1327, 2022.
    
    \bibitem{niu2023characterizations}
    Tailin Niu, Kangquan Li, Longjiang Qu, and Chao Li.
    \newblock Characterizations and constructions of $n$-to-1 mappings over finite fields.
    \newblock {\em Finite Fields and Their Applications}, 85:102126, 2023.
    
    \bibitem{nyberg1993differentially}
    Kaisa Nyberg.
    \newblock Differentially uniform mappings for cryptography.
    \newblock In {\em Workshop on the Theory and Application of Cryptographic Techniques}, pages 55--64. Springer, 1993.
    
    \bibitem{peng2017new}
    Jie Peng, Chik~How Tan, and Qichun Wang.
    \newblock New secondary constructions of differentially 4-uniform permutations over $\gf_{2^{2k}}$.
    \newblock {\em International Journal of Computer Mathematics}, 94(8):1670--1693, 2017.
    
    \bibitem{pott2010switching}
    Alexander Pott and Yue Zhou.
    \newblock Switching construction of planar functions on finite fields.
    \newblock In {\em International Workshop on the Arithmetic of Finite Fields}, pages 135--150. Springer, 2010.
    
    \bibitem{qin2024new}
    Xiaoer Qin and Li~Yan.
    \newblock New results on $n$-to-1 mappings over finite fields.
    \newblock {\em Finite Fields and Their Applications}, 98:102469, 2024.
    
    \bibitem{qu2016more}
    Longjiang Qu, Yin Tan, Chao Li, and Guang Gong.
    \newblock More constructions of differentially 4-uniform permutations on $\gf_{2^{2k}}$.
    \newblock {\em Designs, Codes and Cryptography}, 78:391--408, 2016.
    
    \bibitem{qu2013constructing}
    Longjiang Qu, Yin Tan, Chik~How Tan, and Chao Li.
    \newblock Constructing differentially 4-uniform permutations over $\gf_{2^{2k}}$ via the switching method.
    \newblock {\em IEEE Transactions on Information Theory}, 59(7):4675--4686, 2013.
    
    \bibitem{qu2025parametric}
    Longjiang Qu, Qiancheng Zhang, and Kangquan Li.
    \newblock Parametric construction approach of balanced boolean functions from two-to-one mappings.
    \newblock {\em Journal of Cryptology}, 38(3):1--42, 2025.
    
    \bibitem{rai2025permutation}
    Amritanshu Rai and Rohit Gupta.
    \newblock Permutation polynomials of the form $x^rh(x^{q-1})$ over $\gf_{q^2}$ with even characteristics.
    \newblock {\em Finite Fields and Their Applications}, 104:102594, 2025.
    
    \bibitem{wan1998characteristic}
    Zhe-Xian Wan.
    \newblock A characteristic property of self-orthogonal codes and its application to lattices.
    \newblock {\em Bulletin of the Belgian Mathematical Society-Simon Stevin}, 5(2/3):477--482, 1998.
    
    \bibitem{williams1975note}
    Kenneth~S Williams.
    \newblock Note on cubics over {$ {\rm GF}(2^n)$ and ${\rm GF}(3^n)$}.
    \newblock {\em Journal of Number Theory}, 7(4):361--365, 1975.
    
    \bibitem{wu2017permutation}
    Danyao Wu, Pingzhi Yuan, Cunsheng Ding, and Yuzhen Ma.
    \newblock Permutation trinomials over $\gf_{2^m}$.
    \newblock {\em Finite Fields and Their Applications}, 46:38--56, 2017.
    
    \bibitem{wu2021new}
    Yanan Wu, Nian Li, and Xiangyong Zeng.
    \newblock New {PcN} and {APcN} functions over finite fields.
    \newblock {\em Designs, Codes and Cryptography}, 89:2637--2651, 2021.
    
    \bibitem{xie2022new}
    Xi~Xie, Bing Chen, Nian Li, and Xiangyong Zeng.
    \newblock New classes of bent functions via the switching method.
    \newblock In {\em International Workshop on the Arithmetic of Finite Fields}, pages 298--309. Springer, 2022.
    
    \bibitem{xu2016constructing}
    Guangkui Xu, Xiwang Cao, and Shanding Xu.
    \newblock Constructing new {APN} functions and bent functions over finite fields of odd characteristic via the switching method.
    \newblock {\em Cryptography and Communications}, 8:155--171, 2016.
    
    \bibitem{yuan2021twoandinvolution}
    Mu~Yuan, Dabin Zheng, and Yanping Wang.
    \newblock Two-to-one mappings and involutions without fixed points over $\gf_{2^n}$.
    \newblock {\em Finite Fields and Their Applications}, 76:101913, 2021.
    
    \bibitem{zhang2025balancedbooleanfunctionsfewvalued}
    Qiancheng Zhang, Kangquan Li, and Longjiang Qu.
    \newblock Balanced {Boolean} functions with few-valued {Walsh} spectra parameterized by {$P(x^2+x)$}.
    \newblock {\em arXiv preprint arXiv:2506.19521}, 2025.
    
    \bibitem{zheng2024many}
    Yanbin Zheng, Yanjin Ding, Meiying Zhang, Pingzhi Yuan, and Qiang Wang.
    \newblock On many-to-one mappings over finite fields.
    \newblock {\em arXiv preprint arXiv:2408.04218}, 2024.

\end{thebibliography}
\end{document}